\documentclass [11pt]{article}
\usepackage{amsmath,amsthm,amsfonts,amscd,eucal,latexsym,amssymb,mathrsfs,cancel,color}
\usepackage{epsfig}  
\usepackage{simplewick,bm}
\usepackage{hyperref}
\usepackage{tikz}
\usetikzlibrary{decorations.pathmorphing}
\usetikzlibrary{decorations.markings}
\usepackage{tikz-feynman}
\tikzfeynmanset{compat=1.1.0}
\usepackage{adjustbox}
\input xy
\xyoption{all}
 
\def\beq{\begin{eqnarray}}
\def\eeq{\end{eqnarray}}

\def\la{\lambda}

\definecolor{NiColor}{RGB}{77,77,255}

\newcommand{\Fmuc}[1][]{\mathscr{F}_{\mu\textrm{c}_{#1}}}

\def\Fmuc{\mathcal{F}}

\newtheorem{theorem}{Theorem}[section]

\newtheorem{proposition}{Proposition}[section]

\theoremstyle{definition}

\begin{document}

%
%
 
\par 
\bigskip 
\LARGE 
\noindent 
{\bf Equilibrium states in Thermal Field Theory and in Algebraic Quantum Field Theory.} 
\bigskip 
\par 
\rm 
\normalsize 
 

\large
\noindent 
{\bf Jo\~ao Braga de G\'oes Vasconcellos$^{1,a}$}, {\bf Nicol\`o Drago$^{2,3,b}$}, {\bf Nicola Pinamonti$^{1,4,c}$},  \\
\par
\small
\noindent$^1$ Dipartimento di Matematica, Universit\`a di Genova - Via Dodecaneso, 35, I-16146 Genova, Italy. \smallskip

\noindent$^2$ Dipartimento di Matematica, Universit\`a di Trento - Via Sommarive, 14, I-38050 Povo (TN), Italy. \smallskip

\noindent$^3$ Istituto Nazionale di Fisica Nucleare - TIFPA via Sommarive, 14, I-38123 Povo (TN), Italy. \smallskip

\noindent$^4$ Istituto Nazionale di Fisica Nucleare - Sezione di Genova, Via Dodecaneso, 33 I-16146 Genova, Italy. \smallskip

\smallskip

\noindent E-mail: 
$^a$braga.vasconcellos@dima.unige.it, 
$^b$nicolo.drago@unipv.it, 
$^c$pinamont@dima.unige.it\\ 

\normalsize

\par 
 
\rm\normalsize 

\rm\normalsize 
 
 
\par 
\bigskip 
 
\rm\normalsize 
\noindent {\small Version of \today}

\par 
\bigskip

\rm\normalsize

\bigskip

\noindent 
\small 
{\bf Abstract}
In this paper we compare the construction of equilibrium states at finite temperature for self-interacting massive scalar quantum field theories on Minkowski spacetime proposed by Fredenhagen and Lindner in \cite{FredenhagenLindner} with results obtained in ordinary thermal field theory, by means of real time and Matsubara (or imaginary time) formalisms. 
In the construction of this state, even if the adiabatic limit is considered, the interaction Lagrangian is multiplied by a smooth time cut-off.
In this way the interaction starts adiabatically and the correlation functions are free from divergences. 
The corresponding interaction Hamiltonian is a local interacting field smeared over the interval of time where the chosen cut-off is not constant. We observe that, in order to cope with this smearing, the Matsubara propagator which is used to expand the relative partition function between the free and interacting equilibrium states, needs to be modified. 
We thus obtain an expansion of the correlation functions of the equilibrium state for the interacting field as a sum over certain type of graphs with mixed edges, some of them correspond to modified Matsubara propagators and others to propagators of the real time formalism.
An integration over the adiabatic time cut-off is present in every vertex.
However, at every order in perturbation theory, the final result does not depend on the particular form of the cut-off function.  
The obtained graphical expansion contains in it both the real time formalism and the Matsubara formalism as particular cases. 
For special interaction Lagrangians, the real time formalism is recovered in the limit where the adiabatic start of the interaction occurs at past infinity. At least formally, the combinatoric of the Matsubara formalism is obtained in the limit where the switch on is realized with an Heaviside step function and the field observables have no time dependance. 
Finally, we show that a particular factorisation which is used to derive the ordinary real time formalism holds only in special cases and we present a counterexample. We conclude with the analysis of certain 
correlation functions and we notice that corrections to the self-energy in a $\lambda\phi^4$ at finite temperature theory are expected.
\normalsize

\vskip .3cm

\section{Introduction}
Thermal field theory, seen as a combination of quantum field theory and statistical mechanics, has been founded more than fifty years ago with the pioneering works of Matsubara \cite{Mat55}, Keldysh \cite{Kel65} and Schwinger \cite{Sch61} (see also \cite{NS84}, \cite{vH86}).
At that time, one of the main motivations was the analysis of phase transitions of hadronic matter predicted to occur at high temperature in quantum cromodynamics \cite{McL85}. The formalism of thermal field theory has also been also applied to the analysis of baryogenesis in the early-universe (see e.g. \cite{Ka84, CH88, ProkopecA, ProkopecB}) and to derive transport equations for a system of quantum fields like in the study of the dynamics of quark-gluon plasma  (for a more recent application, see \cite{CM05} for instance). 

\bigskip
Extensive reviews of the methods of thermal field theory can be found in the book of Le Bellac \cite{LeBellac} or in the book of Kapusta and Gale \cite{Kapusta}. The connection of thermal field theory with rigorous quantum statistical mechanics can be found in the review of Landsman and van Weert \cite{Landsman}.
Here we would like to briefly recall the standard picture of thermal field theory which consists in the analysis of a 
self-interacting quantum theory at finite temperature. Self interaction is treated by means of perturbation theory, hence interacting fields are represented with free fields at finite temperature.
It is furthermore necessary to consider the modifications of the equilibrium state due to interaction. 
This is usually done in the following way.
Let us start with a free quantum field theory, initially prepared in a free equilibrium state at a given temperature. At a certain time $t_i$, the perturbative interaction is switched on, both the state and the dynamics are perturbed, and the system is left to evolve. 
In particular, the $n-$point Green functions for points located in the future of $t_i$
 are formally
\begin{equation}\label{eq:corr-path}
	G(x_1, \dots,  x_n) = \left. \frac{1}{Z(0)}\frac{\delta^n}{\delta j(x_1)\dots  \delta j(x_n)} \bigg< T_C \exp i \int_C dx \, \mathcal{L}_I (x) + j(x) \phi(x)\bigg> \right\vert_{j = 0},
\end{equation}
where brackets correspond to averaging with respect to the initial free equilibrium state at a given inverse temperature $\beta$ and $j$ is a complex-valued source. 
The time ordering $T_C$ and the integral of the interaction lagrangian $\mathcal{L}_I$ over the time coordinate $x^0$ is done along a contour $C$ in the complex plane.
The standard contour which is used is known as Keldish-Schwinger contour. This contour has usually three branches, it starts at $t_i$ goes up to $t_f$ then back to $t_i$  and the final branch goes from $t_i$ to $t_i - i\beta$. The first two branches take into account the perturbation of the fields in the interaction picture. It is thus necessary to chose a $t_f$ which is larger than the time coordinates of every point $x_i$.
The last branch is over imaginary time and it is necessary to modify the free equilibrium state into an interacting equilibrium state and thus it corresponds to the relative partition function constructed with the free and interacting dynamics.
	 
It is usually very difficult to derive a set of Feynman rules for the perturbative representation of correlation functions in \eqref{eq:corr-path}.
However, if one is interested in computing the relative partition function only, it is possible to discard the first two branches of the contour $C$ and consider an integration from $t_i$ directly to $t_i-i\beta$. 
This analysis is named {\bf Matsubara} or imaginary time formalism. The very same analysis holds if one is interested in computing correlation functions at imaginary times.  
Notice that $n$-point functions with all the time arguments along this vertical complex contour are correlation functions of an Euclidean field theory. Arguably, one of the greatest advantages of the Matsubara, or imaginary time formalism lies on the fact that the set of Feynman rules obtained with such analysis are quite similar to those of a field theory
in a vacuum state.
However, to obtain Green functions in Minkowski spacetime, it is necessary to perform a Wick rotation of Euclidean correlation functions, a task which is usually not trivial, see \cite{OS75a,OS75b} for details about the analytic continuation of correlation functions. 	

Another possibility starts from the observation that in some cases, in the limit where $t_i\to -\infty$ or equivalently where all points $x_j$ in \eqref{eq:corr-path} are rigidly translated to future infinity, the contribution of the last branch, from $t_i$ to $t_i-i\beta$ factorizes from the correlation functions. If this is the case the integration over the imaginary part of the contour can be discarded. The computational scheme obtained in this way is known as {\bf real-time} formalism.

\bigskip
As already said above, for a generic contour $C$ where all the branches are not discarded, it is not easy to obtain a set of Feynman rules for the perturbative representation of \eqref{eq:corr-path}. 
Furthermore, since the interaction is switched on instantaneously, some divergences in certain correlation functions are expected.
This problem could be overcome using an adiabatic or smooth switching on function. 
Unfortunately, in this way the interaction Hamiltonian is not local in time, hence the integration over the imaginary part of the path turns out to be problematic. Moreover, in the limit where $t_i\to -\infty$ the contribution of relative partition function, (the integration of the last branch from $t_i$ to $t_i-i\beta$) factorizes only in very special cases, see \cite{DFP-18} for further details. 
Finally, as we shall see in the loop computation of section \ref{sec:practical}, the real time formalism is very sensitive to the way in which the interaction is switched on. 

\bigskip
Recently, another approach to the perturbative treatment of quantum field theories at finite temperature has been developed within the framework of perturbative algebraic quantum field theory \cite{BF00, BDF09, BFV03, HW01, HW02, FredenhagenRejzner}. 
In this formalism, Fredenhagen and Lindner \cite{FredenhagenLindner} and Lindner in \cite{Lindner} have shown how it is possible to construct equilibrium states (KMS states at finite temperature) for an interacting quantum field theory in the adiabatic limit. 
In those works, the authors used in the realm of quantum field theories, ideas proper of rigorous quantum statistical mechanics as, for example, the construction of equilibrium states for a perturbed dynamics studied by Araki in \cite{Araki}.
To make this connection clear we recall that the starting point of the construction of equilibrium states in rigorous quantum statistical mechanics \cite{BR1,BR} is a $C^*-$dynamical system $(\mathfrak{A},\tau_t)$, where  $\mathfrak{A}$ is the $C^*-$algebra generated by the observables of the theory and  where $\tau_t$ is the one-parameter group of $*-$automorphisms which implements the time evolution. 
Equilibrium states at inverse temperature $\beta$ are the normalized positive functionals over $\mathfrak{A}$ which satisfy the Kubo Martin Schwinger (KMS) condition. We recall here that a state $\Omega$ satisfies the KMS condition at inverse temperature $\beta > 0$, also called $\beta-$KMS, with respect to the time evolution $\tau_t$, if, for every $A,B\in\mathfrak{A}$, $F(t):=\Omega(A\tau_t(B))$ is analytic in the domain $\text{Im}(t)\in (0,\beta)$, bounded on its closure, and such that $F(t+i\beta)=\Omega(\tau_t(B)A)$. 

KMS states enjoy another important property, in particular they are stable under perturbation of the dynamics.
We consider for simplicity the case where 
$\mathfrak{A}$ is represented as a set of bounded operator over some Hilbert space $\mathcal{H}$ and, in this representation,  $\tau_t$ is generated by an Hamiltonin $H$ which is a (possibly unbounded) selfadjoint operator over $\mathcal{H}$. Let $H_I$ be a selfadjoint element of $\mathfrak{A}$, then the dynamics perturbed by $H_I$ is the one generated by $H+H_I$. We denote the corresponding one-parameter group of $*-$automorphisms by $\tau_t^{H_I}$. 
A given $\beta-$KMS state $\Omega$ with respect to $\tau_t$ is not a $\beta-$KMS state with respect to $\tau_t^{H_I}$. However, if a certain clustering condition is satisfied, see e.g. \cite{BKR} for further details, the following stability result, also called return to equilibrium, holds
\begin{align}\label{eq: Cstar return-to-equilibrium}
\lim_{t\to\infty} \Omega(\tau_t^{H_I}(A)) = \Omega^{H_I}(A)\,, \qquad A\in\mathfrak{A},
\end{align}
where $\Omega^{H_I}$ is a $\beta-$KMS with respect to $\tau_t^{H_I}$ \cite{BKR, BR, Araki}. Furthermore, 
\begin{equation}\label{eq:intstate} 
\Omega^{H_I}(A) = \frac{\Omega(AU(i\beta))}{\Omega(U(i\beta))}, \qquad A\in\mathfrak{A},
\end{equation}
where $U(t)$ is the cocycle intertwining the free and perturbed time evolution and where 
the function $t\mapsto \Omega(AU(t))$ is evaluated at $i\beta$.
Actually as shown by Araki in \cite{Araki}, this function can be continued to a continuous function in the strip $\text{Im}(t)\in[0,\beta]$ which is analytic in the interior.
Hence, if $\Omega$ has the desired clustering property, in view of 
\eqref{eq: Cstar return-to-equilibrium} we might construct equilibrium states for the perturbed dynamics either taking the large time limit of $\Omega\circ \tau_t^{H_I}$ or with \eqref{eq:intstate}.
Notably, if the stability does not hold, \eqref{eq:intstate} still furnishes a well-defined equilibrium state.
We observe that if $e^{-\beta H}$ is a trace class operator over a suitable Hilbert space the KMS state reduces to the ordinary Gibbs state. In that case, formula \eqref{eq:intstate} for the perturbed state gives for $A\in\mathfrak{A}$
\[
\Omega(A) = \frac{\text{Tr}(A e^{-\beta H} )}{\text{Tr}(e^{-\beta H} )}, \qquad
\Omega^{H_I}(A) =  \frac{\Omega(A\rho)}{\Omega(\rho)}, \qquad \rho = e^{-\beta(H+{H_I})}e^{\beta H}.
\]
We thus conclude that $U(i\beta)$ in \eqref{eq:intstate} plays the role of the relative partition function $\rho$ also when $e^{-\beta H}$ is not a trace class operator.

Fredenhagen and Lindner in \cite{FredenhagenLindner} managed to extend the Araki formalism, and in particular \eqref{eq:intstate}, to the case of a self-interacting massive scalar field $\phi$ propagating over a Minkowski spacetime where the interaction is treated perturbatively and where the adiabatic limit is considered.
In particular, they noticed that even if the Hamiltonian is not at disposal, it is still possible to obtain in the field algebra the generator $K$ of the cocycle intertwining the free and the interaction time evolution, at least when the interaction Lagrangian  $\chi h \mathcal{L}_I$ is made of compact support with a time cut-off $\chi$ and a space cut-off $h$. That generator plays the role of the perturbation Hamiltonian and furthermore
The Green functions of the interacting theory can be obtained  
\begin{equation}\label{eq:state-fl}
G(x_1,\dots, x_n)=
\sum_{l\geq 0} (-1)^l\!\!\!\int\limits_{0<u_1<\dots <u_l<\beta}\!\!\!\!\!\! du_1 \dots du_l\; \omega^{\beta,c}(R_V(T(\phi(x_1),\dots,\phi(x_n)))\otimes K_{iu_1}\otimes\dots \otimes K_{iu_l} )\, ,
\end{equation}
where $\omega^{\beta,c}$ are the connected functions of $\omega^\beta$ which is the equilibrium state at inverse temperature $\beta$ of the free theory over which perturbation theory is developed; $R_V$ is the Bogoliubov map, see e.g. \eqref{eq:bogoliubov}, $K=R_V(\int \dot\chi h \mathcal{L}_I)$ is the interacting Hamiltonian  and
$K_{iu}$ its imaginary  time translation.  Furthermore, in \cite{FredenhagenLindner} it is shown that the limit $h\to 1$ where the spatial cutoff is removed can be taken. The obtained state, denoted by $\omega^{\beta,V}$, 
see also \eqref{eq:KMS-state} and \eqref{Eqn: cluster expansion of interacting KMS state}, does not depend on $\chi$ anymore and it satisfies the KMS condition with respect to the interacting time evolution. Further details about this construction will be given in section \ref{sec:paqft}.
Additional properties of the state constructed in \cite{FredenhagenLindner} have been discussed in \cite{DHP,DFP-18,DFP-19}.
Notice that, in general, the state $\omega^{\beta,V}$ satisfies the KMS condition with respect to the interacting time evolution $\alpha_t^V$ only if the correction to $\omega^{\beta}$ described by the elements of the sum present in \eqref{eq:state-fl} for $l>1$ are taken into account.

We stress that it is usually necessary to consider the state $\omega^{\beta,V}$ (analogous to \eqref{eq:intstate}), rather than considering the large time limit of $\omega^\beta\circ\alpha^V_t$. Actually, in  \cite{DFP-18}, it is shown that for large times $t\to+\infty$ the two approaches coincides -- namely $\omega^\beta\circ\alpha^V_t\to\omega^{\beta,V}$ -- only in special cases.
In particular, in a $g\phi^4$ theory or in a $g\phi^2$ theory, they do not coincide because, in that case, for generic observables $A$ and $B$ 
\[
\lim_{t\to\infty} \lim_{h\to1} \omega^{\beta,c}(\alpha_t^V(A) \otimes B) \neq 0
\]
namely the clustering condition used to prove \eqref{eq: Cstar return-to-equilibrium} is not satisfied.
We furthermore notice that, only when this clustering condition holds, we can construct an equilibrium state with the large time limit of $\omega^{\beta}\circ \alpha_t^V$. This correspond to discard 
 the integration over the imaginary branch of the Keldish-Schwinger contour and moving $t_i$ to minus infinity. 
 
To support this observation, 
consider the quasifree KMS state with respect to a fixed Minkowski time at inverse temperature $\beta$ of a free quantum scalar field theory of mass $m$. Its two-point function is translation invariance and it is equal to 
\begin{equation}\label{eq:2pt-thermal-m}
\Delta^{\beta,m}_+(x) =  \frac{1}{(2\pi)^3}\int d^3\mathbf{p} e^{i \mathbf{p} \cdot \mathbf{x} }  \left( \frac{1}{1-e^{-\beta w_m }} \frac{e^{-i w_{m} x^0 }}{2w_{m} }   +  \frac{1}{e^{\beta w_m}-1}  \frac{e^{i w_{m} x^0 }}{{2w_{m} }}   \right)
\end{equation}
where $w_m = \sqrt{|\mathbf{p}|^2+m^2}$.
 Notice that, thanks to the validity of the principle of perturbative agreement \cite{HW05, DHP},
 we can treat changes of the mass $m$ to $\tilde{m}=\sqrt{m^2+\lambda}$ in two ways. Either exactly substituting the mass 
$m$ with $\tilde{m}$ in the various propagators appearing in the theory or with perturbation theory, namely using an interaction Lagrangian of the form $\lambda\phi^2$. 
We observe that, \cite{DHP, DappiaggiDrago,  Drago} 
 perturbation theory alone is able to change the mass in the modes $\frac{e^{-i \omega_{m} x^0 }} {\sqrt{2\omega_{m}}}$ appearing in \eqref{eq:2pt-thermal-m} but it is not able to change the mass in the Bose factors $\pm(e^{\pm \beta w_m}-1)^{-1}$.

\bigskip

In this paper we would like to compare the traditional methods used in thermal field theory with the construction proposed by Fredenhagen and Lindner. We shall see that an expansion in terms of graph of the correlation functions \eqref{eq:state-fl} in the state studied by Fredenhagen and Lindner can be obtained, see e.g. section \ref{sec:graph-exp}.
In particular, we show in Proposition \ref{pr:real-time-formalism} that the Bogoliubov map $R_V$ present in \eqref{eq:state-fl} can be expanded in graphs whose vertices correspond to the propagators given in \eqref{eq:real-time-propagator}. These graphs are very similar to those appearing in the real time formalism.
To obtain a more explicit graphical description of the connected correlation functions present in the elements of the sum in \eqref{eq:state-fl} for $l>1$, we study the analytic continuation of the thermal propagator in the imaginary domain and its expansion in terms of Matsubara frequencies, see section \ref{sec:propagator}, but keeping the time with a non vanishing real part. This is necessary to properly treat the adiabatic smearing in time of the interaction Lagrangian. The obtained propagator is a modified version of the Matsubara one -- \textit{cf.} equation \eqref{eq:DeltaHat}.
Thanks to the symmetry satisfied by this propagator, we show in Proposition \ref{pr:Fsymmetric} and in Theorem 
\ref{th:matsubara-frequency-conservation} that the integration in the $u-$variables in \eqref{eq:state-fl} can be taken ensuring conservation of the sum of Matsubara frequencies at every interaction Hamiltonian $K$. The combinatoric of the expansion of the correlation function is similar to that of the imaginary time formalism.

We summarize the computational rules for the obtained graphical expansion in section \ref{sec:graph-exp}.
These rules combine both the real time formalism used to evaluate the action of the interactions on local fields (section \ref{sec:graph-rt}) and the Matsubara formalism used to evaluate the modification of the state, necessary to have an equilibrium state for the interacting theory (section \ref{sec:graph-matsubara}).
If return to equilibrium holds (analogous to \eqref{eq: Cstar return-to-equilibrium} see also \eqref{eq:return-to-equilibrium}), the evaluation of the correlation functions in the state studied by Fredenhagen and Linder reduces to the analysis done in the ordinary real time formalism.
Unfortunately, we are able to prove return to equilibrium only if a certain clustering condition holds (see \eqref{eq:clustering-condition}) and this is typically not the case when the adiabatic limit is considered, that is, when the spatial support of the interaction Lagrangian is non-compact, which is actually the standard case in Quantum Field Theory \cite{DFP-18}.
In this limit, the state studied by Fredenhagen and Linder still furnishes meaningful results while the real time formalism in some cases needs to be amended (see the discussion present at the end of section \ref{sec:correction-phi2}).
Moreover, as we shall see in section \ref{sec:practical}, perturbative computations at first loop order done in the real time formalism depend very much on the time in which the interaction is switched on.
In particular, at the end of section \ref{sec:correction-phi2} we argue that additional contributions to the self-energy arise if a theory with $g\phi^4$-interaction is considered.
Indeed, due to the presence of the thermal mass, both the real time formalism and the imaginary time formalism need to be combined in order to provide the complete result.
The form of the self-energy up to two loop order in  
real time formalism has been obtained by \cite{Altherr} in the analysis of infrared problems present in the limit of vanishing mass. Similar computations in the imaginary time formalism can be found in \cite{Parwani}.

\bigskip
The paper is organised as follows.
In section \ref{sec:paqft} we briefly introduce the formalism of algebraic quantum field theory in perturbation theory and we recall further details about the construction of equilibrium state for the interacting theory.
In section \ref{Sec: Expansion of the perturbative series} we present the main theorems which permit to obtain the graphical expansion of the correlation functions of the interacting equilibrium states given by Fredenhagen and Lindner.
Moreover, we discuss how this graphical expansion contains both the real time and the Matsubara formalisms as sub-cases.
Section \ref{sec:practical} contains some examples of the evaluation of correlation functions at fixed finite perturbation order.
These examples are chosen in order to highlight  the differences with the ordinary formalisms with particular attention to the discussion of corrections to the self-energy.

\section{Introduction to pAQFT and equilibrium states}\label{sec:paqft}

Fredenhagen and Lindner in \cite{FredenhagenLindner}  have extended the construction of equilibrium states for perturbed dynamics to the case of self-interacting quantum field theories. This has been achieved using a formula similar to \eqref{eq:intstate} adapted to perturbation theory.
In order to recall this result we briefly introduce the AQFT framework, which we will exploit henceforth -- \textit{cf}. \cite{BDF09,HW01, HW02, FredenhagenRejzner}.

We consider an interacting real scalar field theory
propagating over a four dimensional Minkowski spacetime $M$ equipped with a metric $\eta$ whose signature is $(-,+,+,+)$.
The Lagrangian of the theory we want to consider is of the form
\[
\mathcal{L} =\mathcal{L}_0 - \mathcal{L}_I =    -\frac{1}{2}\partial_\mu \phi\partial^\mu \phi - \frac{m^2}{2} \phi^2 - \mathcal{L}_I,
\]
where $m$ is a non vanishing mass and $\mathcal{L}_I$ is the interaction Lagriangian.
 In particular, for $\mathcal{L}_I=0$ the field is a free massive scalar field theory. 
The equation of motion associated to  $\mathcal{L}$ is of the form
\[
P\phi - V' = \Box \phi - m^2 \phi - V'=0
\]
where $\Box$ is the d'Alembert operator of the considered Minkowski spacetime. 
$V'$ is the first functional derivative of the difference of the interacting and free local action. An example of the considered $V$ is the following   
\[
V=\int_M \mathcal{L}_I d\mu_x =   \lambda \int_M \frac{\phi^n(x)}{n} d\mu_x, \qquad V'(x) = \lambda \phi^{n-1}(x). 
\]

The observables of the corresponding quantum theory are obtained as elements of the algebra $\mathcal{A}$ generated by all possible normal ordered abstract Wick polynomials $\mathcal{W}:=\{\Phi^n(f),\partial^\alpha\Phi^\beta(f)\dots \}$, namely all possible local fields smeared with smooth compactly supported test functions.  The product among these objects has to encode Einstein causality, furthermore, the commutation relations obtained with this product at linear order in $\hbar$ for $\hbar$ which tends to $0$ must reduce to $i$ times the Peierls brackets of the corresponding classical theory \cite{FredenhagenRejzner}.

A direct construction of this product -- and thus of the corresponding algebra  $\mathcal{A}_{\text{free}}$ -- is available for the case of linear theories ($\mathcal{L}_I=0$). Interacting fields are then constructed as formal power series with coefficients in $\mathcal{A}_{\text{free}}$, with the caveat that elements of $\mathcal{A}_{\text{free}}$ are considered to be off-shell.

A concrete representation of $\mathcal{A}_{\text{free}}$ can be given in terms of functionals over {\bf off-shell smooth field configurations} $\phi\in \mathcal{E}:=C^\infty(M)$.
Hence elements of $\mathcal{A}_{\text{free}}$ are seen as elements of the set of {\bf microcausal functionals} over off-shell smooth field configurations 
\[
\Fmuc = \left\{\left.F:\mathcal{E}\to\mathbb{C}\right|\; F\; \textrm{compactly supported},\; \text{smooth},\;  \text{microcausal} \right\}\,.
\]
Here a functional  $F$ is said to be smooth and compactly supported if all its functional derivatives $F^{(n)}$ exist as compactly supported distributions -- that is, $F^{(n)}\in\mathcal{E}'(M)$ for all $n\in\mathbb{N}$.
The compactness of $\operatorname{supp}(F)$ implies that it exists a compact subset $K\subset M$ such that $\text{supp}(F^{(n)})\subset K^n$ for all $n\in\mathbb{N}$.
Finally, $F$ is said to be microcausal if the wave front set \cite{Hormander} of its functional derivatives is such that $\text{WF}(F^{(n)}) \cap (\overline{V^+}^n\cup \overline{V^{-}}^n)=\emptyset$, where $\overline{V^+}^n$ is the subset of $T^*M^n$ formed by elements $(x_1,\dots ,x_n;k_1,\dots, k_n)$  where $k_i$ are future-directed causal covectors and similarly $\overline{V^-}^n$ is formed by past-directed covectors.
This rather technical condition will be necessary to ensure the well definiteness of the products we will introduce below \cite{HW01} -- \textit{cf.} equation \eqref{eq:star-product-abstract}.

Relevant subsets of $\Fmuc$ are the set of {\bf regular functionals} $\mathcal{F}_\text{reg}$ formed by functionals whose functional derivatives are smooth compactly supported functions. The other relevant subset is the set of {\bf local functionals}, namely 
\[
\mathcal{F}_{\text{loc}}= \left\{F \in \Fmuc \left|\; \text{supp}F^{(n)}\subset D_n\subset M^n \right.\right\}
\] 
where $D_n$ is the diagonal of  $M^n$, namely the image of $M$ under the map $\iota:M\to M^n$, $\iota(x) := (x,x, \dots , x)$. 
In the case of free theories the product in $\Fmuc$ is such that
\begin{equation}\label{eq:star-product-abstract}
	A\star_w B = \sum_{n\geq 0} \frac{\hbar^n}{n!} \left\langle A^{(n)},(w)^{\otimes n}  B^{(n)} \right\rangle_n
\end{equation}
where $w$ is any Hadamard two-point function associated to the free massive Klein-Gordon equation $P\phi=0$ where $P=\Box-m^2$.
By definition this means that $w \in \mathcal{D}'(M^2)$ satisfies the following properties: (a) $w\circ(1\otimes P)$ and $w\circ (P\otimes 1)$ are smooth functions; (b) for every $f,g\in\mathcal{D}(M)$, $\overline{w(f,g)} = w(\overline{g},\overline{f})$ and $w(f,g)-w(g,f) = i \Delta(f,g)$ where $\Delta=\Delta_R-\Delta_A$ is the causal propagator of the free theory, namely the retarded minus advanced fundamental solutions of the massive Klein-Gordon equation $P\phi=0$; (c) $w$ satisfies the {\bf microlocal spectrum condition} \cite{Radzikowski, BFK}
\[
\text{WF}(w)=\left\{ (x_1,x_2;k_1,k_2)\in T^*M\setminus\lbrace 0\rbrace |\; (x_1,k_1)\sim (x_2,-k_2), k_1\triangleright 0   \right\}
\]
where $(x_1,k_1)\sim (x_2,k_2)$ holds true if $x_1$ and $x_2$ are joined by a null geodesic $\gamma$, $\eta^{-1}k_1$ is a tangent vector of $\gamma$ at $x_1$ and if $k_2$ is the parallel transport of $k_1$ along $\gamma$.
Finally $k_1\triangleright 0$ holds true if $k_1$ is future directed. 

If we equip $\Fmuc$ with $\star_w$ given in \eqref{eq:star-product-abstract} and with the involution formed by complex conjugation we obtain a representation of the $*-$algebra of the free theory. Up to renormalization freedom, see e.g. \cite{HW01}, this representation is such that
\begin{equation}\label{eq:map-monomials}
	r_w:\mathcal{A}_{\text{free}}\to (\Fmuc,\star_w)\qquad
	r_w (\Phi^n(f)) := \zeta_{w-H} \int f \phi^n d\mu\,.
\end{equation}
Here $H$ is the Hadamard function constructed only with local information.
In the case of a massless free theory over a Minkowski spacetime, $H$ is the two-point function of the vacuum, while for massive free theories an explicit construction can be found in \cite[App. A]{BDF09}.
Furthermore, 
\[
\zeta_{w} F := \exp\left(\frac{1}{2}\int d\mu_xd\mu_y w(x,y)\frac{\delta^2}{\delta \phi(x)\delta \phi(y)}\right) F\,.
\]
Different choices of $w$ produce isomorphic algebras via $r_{w'}\circ r_w^{-1}:=\zeta_{w'-w}\colon(\mathcal{F},\star_w)\to (\mathcal{F},\star_{w'})$ \cite{BDF09}.
However, notice that local non linear functionals are not invariant under this isomorphisms, while for local linear fields we have $r_w (\Phi(f)) = \int f \phi d\mu$ no matter the choice of $w$.
With a little abuse of notation we shall thus often write $r_w (\Phi(f))=\Phi(f)$.
Local functionals in $(\Fmuc,\star_w)$ are understood as being normal ordered fields with respect of $w$. 

\bigskip
To construct interacting fields within perturbation theory we use the Bogoliubov construction.
The starting point is the definition of time-ordered products among local fields \cite{HW02,BF00}.
Once a representation $(\Fmuc,\star)$ is chosen, the time-ordered product is defined as a sequence of maps $T\colon\mathcal{F}_{\mathrm{loc}}^{\otimes n}\to\Fmuc$ from multilocal functionals $\mathcal{F}_{\text{loc}}^{\otimes n}$ to $\Fmuc$ for all $n\in\mathbb{N}$. 
The time-ordered products are totally symmetric and among regular local functionals  are completely characterised by the {\bf causal factorisation property} which says that 
\[
T(F_1,\dots, F_l,G_1,\dots, G_p) = T(F_1,\dots, F_l)\star T(G_1,\dots, G_p), \qquad F_i \gtrsim G_j, \qquad \forall i,j
\]
where $F \gtrsim G$ holds if $J^{+}(\text{supp} F)\cap J^{-}(\text{supp} G) =\emptyset$. 
As shown by Epstein and Glaser in \cite{EG}, time-ordered products among generic local functionals can be constructed employing a recursive construction over the number of factors, see also
\cite{BF00, HW02}. At each recursive step the causal factorisation property permits to define the distributional kernel associated to the multilocal maps $T(F_1, \dots, F_n)$ up to the thin diagonal $D_n\subset M^n$, namely as elements of $\mathcal{D}'(M^n\setminus D_n)$. The extension of these distributions over  $D_n$ can be obtained using scaling limit techniques similar to those used to extend homogeneous distributions \cite{Hormander}. The replacement for the degree of homogeneity is the Steinman scaling degree \cite{Steinmann}. Though the extension is not unique, the freedom in the extension can be restricted (but not completely resolved) with some further requirements \cite{HW01} -- the remaining freedom is the well-known renormalization freedom. 
As discussed in \cite{DFKR}, this method is equivalent to the more standard dimensional regularization.
For our purposes it is enough to know that these maps exist and that the freedom can be controlled.

Once the time-ordered products are constructed we have all the ingredient to built the formal $S-$matrix of a local functional $V$.
This is nothing but the time-ordered exponential of $V$
\[
S(V) =  1+\sum_{n\geq 1}\frac{(-i)^n}{n!}T(\underbrace{V,\dots, V}_n).
\]
Unfortunately the series written above is usually not convergent, hence, $S(\lambda V)$ can be understood as a formal power series in $\lambda$ with coefficients in $\Fmuc$ only, namely as an element of $\Fmuc[[\lambda]]$.

The time-ordered exponential $S(V)$ has a lot of interesting properties which hold in the sense of formal power series. In particular, $S(V)$ is an unitary element of $\Fmuc[[\lambda]]$. With the time-ordered exponential at disposal we may give meaning to the {\bf Bogolibov map}
\begin{equation}\label{eq:bogoliubov}
R_V(A) = S(V)^{-1} \star T(S(V),A), \qquad 
\end{equation}
which is well-defined for all $A\in\bigoplus_{n\geq 0}T\mathcal{F}_{\text{loc}}^{\otimes n}$.

With the Bogoliubov map at disposal we may represent interacting local fields on $\Fmuc[[\lambda]]$. Actually, the Bogoliuobv map applied to a local linear field $\Phi(f)$ satisfies an off-shell version the interacting equation of motion $R_V(\Phi(Pf)) - R_V(V^{(1)}(f))  = \Phi(Pf)$.
See e.g. \cite{BDF09, FredenhagenRejzner} for further details.
The interacting field algebra $\mathcal{F}_I$ is then represented in $\Fmuc[[\lambda]]$ as the algebra generated by $R_V(\mathcal{F}_{\text{loc}})$. The causal factorisation property implies that 
if $F,G\in\mathcal{F}_{\text{loc}}$ with $F\gtrsim G$
\[
R_V(T(F,G)) = R_V(F)\star R_V(G).
\]  
Once a quantum state $\omega$ is chosen, the {\bf generating functional} of the correlation functions in that state is given in terms of the relative $S-$matrix $S_V(J) = S(V)^{-1} \star S(V+J)$
as
\[
\mathcal{Z}(j):= \omega(S_V(J_j)), \qquad J_j:=\int j \phi d\mu 
\]
and thus the correlation functions can be obtained as
\[
G(x_1,\dots, x_n)= \frac{1}{(-i)^n} \left.\frac{\delta^n}{\delta j(x_1)\dots \delta j(x_n)}   \mathcal{Z}(j)\right|_{j=0}.
\]
The state $\omega$ has to be chosen carefully in order to obtain correlation functions analogous to \eqref{eq:corr-path}.

Finally we observe that if the state on the background is quasi-free and its correlation functions are generated by the two-point function $\omega_2$, it is convenient to analyse the expectation values in the representation $(\Fmuc,\star_{\omega_2})$, actually, in that representation
\[
\omega(r_{\omega_2}^{-1}(A)) = \left. A \right|_{\phi=0} = A(0), \qquad  A\in (\Fmuc,\star_{\omega_2}).
\]

\subsection{Time translations and equilibrium states}

Consider a local field $F$, which can be written as
\[
F = \int_M f(x) \Psi(x) d\mu_x,
\]
where $\Psi$ is a section of the jet-bundle of $\phi$ \cite{Kolar-Michor-Slocvak-93}.
The free {\bf time evolution} on a local functional is defined by
\[
\alpha_t(F) =F_t = \int_M f(x-t e_0) \Psi(x) d\mu_x\,,
\]
where $e_0$ is the unit time vector field in the chosen frame.
Since local functionals generate $\Fmuc$, we have that $\alpha_t$ can be extended to $\Fmuc$.
Moreover, if we equip $\Fmuc$ with a $\star_w$ constructed with a translation invariant $w$, $\alpha_t$ is a one parameter group of $*-$automorphisms of $\Fmuc$. 

The {\bf interacting time evolution} is defined by
\[
\alpha_t^V(R_V(F)) = R_V(\alpha_t(F)) = R_V(F_t).
\]
Within this setting we may discuss the same results presented in the context of rigorous quantum statistical mechanics -- \textit{cf.} equations \eqref{eq:intstate} and \eqref{eq: Cstar return-to-equilibrium}.
To construct an equilibrium state for the interacting theory, we need to find the cocycle $U(t)$ which intertwines the free and interacting time evolution,  $\alpha_t^V(\cdot)=U(t)\star \alpha_t^V(\cdot)\star U(t)^{*}$.
Even if the generators of both $\alpha_t$ and $\alpha_t^V$ cannot be obtained as elements of $\Fmuc$, the desired cocycle can be written as a formal power series with coefficients in $\Fmuc$ whenever $V$ is of compact support.
This cocycle has been extensively studied in \cite{FredenhagenLindner}.
In particular,
by the time slice axiom \cite{CF}
a state is characterized by its values on observables supported in an $\epsilon-$neighborhood $\Sigma_\epsilon$ of the $t=0$ Cauchy surface -- specifically, $\Sigma_\epsilon=\{p\in M |\; t(p)<\epsilon \}$.
Therefore, the cocycle needs to be constructed only for observables supported in $\Sigma_\epsilon$. 
In addition, for observables supported in $\Sigma_\epsilon$, using the causal factorisation property, we might restrict the support of the interaction Lagrangian to $\Sigma_{2\epsilon}$. 
Hence, we fix
\begin{equation}\label{eq:Vchih}
V_{\chi h} = \int   \mathcal{L}_I(\tau,\mathbf{x}) \chi(\tau) h(\mathbf{x}) d\tau d^3\mathbf{x}
\end{equation}
where $h\in C^\infty_0(\Sigma)$ is a spatial cut-off which will be eventually removed when the adiabatic limit is considered and $\chi$ is a smooth function of time equal to $1$ on $J^+(\Sigma_\epsilon)\setminus\Sigma_\epsilon$, and equal to $0$ for times smaller than $-2\epsilon$.
Now for $t>0$ and for observables supported in $\Sigma_\epsilon$ the cocycle is 
\[
U(t) = S_V(V_t-V).
\]
With this cocycle at disposal, we can use a formula similar to \eqref{eq:intstate} to write the KMS state for the interacting theory once the equilibrium state $\omega_\beta$ of the free theory with respect to the free time evolution at inverse temperature $\beta$ is given. 
Actually, in \cite{FredenhagenLindner} it has been shown that for every $A\in\mathcal{F}$ the map $t\mapsto\frac{\omega^{\beta}(A\star U(t))}{\omega^{\beta}(U(t))}$ can be continued to a continuous function the strip $\operatorname{Im}(t)\in[0,\beta]$ which is analytic in the interior so that the evaluation
	\begin{equation}\label{eq:KMS-state}
		\omega^{\beta,V}(A) := \frac{\omega^{\beta}(A\star U(i\beta))}{\omega^{\beta}(U(i\beta))}\,,
	\end{equation}
defines a KMS state for the interacting time evolution at inverse temperature $\beta$.
The state $\omega^{\beta,V}$ can actually be expanded as follows
\begin{align}\label{Eqn: cluster expansion of interacting KMS state}
	\omega^{\beta,V}(A)=
	\sum_{n\geq 0} (-1)^n\int_{\beta S_n} du_1 \dots du_n \omega^{\beta,c}(A\otimes \alpha_{iu_1}K\otimes\dots \otimes \alpha_{iu_n}K )\, ,
\end{align} 
where the domain of integration is the symplex $\beta S_n = \{(u_1,\dots ,u_n)| 0<u_1<\dots <u_n <\beta\}$, while
$\omega^{\beta,c}$ denotes the connected correlation function associated with the state $\omega^\beta$.
$K$ is the generator of the cocycle $U(t)$ and it is actually $K=R_V(\dot{V})$ where $\dot{V}= \int \dot\chi h \mathcal{L}_I$ -- \textit{cf.} \cite{FredenhagenLindner}. 

If both $h$ and $\chi$ are of compact support, $V=V_{\chi,h}$ given in \eqref{eq:Vchih} is an element of $\mathcal{F}_{\text{loc}}$. However, in \cite{FredenhagenLindner} it has been shown that the {\bf adiabatic limit} can be taken, namely the limit $h\to1$ of the obtained correlation functions exists and it is finite.
This holds thank to the fast (exponential) decay of the free thermal connected $n$-point functions $\omega^{\beta,c}$ for large spatial separation.
Furthermore, the expectation values obtained in that limit do not depend on the form of the function $\chi$ used to smoothly switch on the interaction -- \textit{cf.} \cite{FredenhagenLindner,DHP}.
In this way an equilibrium state at inverse temperature $\beta$ for a self-interacting scalar field is obtained.  

\section{Expansion of the perturbative series for $\omega^{\beta, V}$ in terms of graphs}\label{Sec: Expansion of the perturbative series}

We shall now discuss a convenient graphical expansion of the $\omega^{\beta,V}$.
To this avail we first present an expansion of the two-point function $\Delta_+^\beta$ -- \textit{cf.} equation \eqref{eq:thermal-2pt} -- of the equilibrium state $\omega^\beta$ in terms of the Matsubara frequencies -- \textit{cf.} equation \ref{eq:thermal-propagator}.
We then proceed in the analysis of the graphical expansion of $\omega^{\beta,V}$ in two steps.
At first we consider the graphical expansion of the Bogolibov map \eqref{eq:bogoliubov}, discussing its relation with the real time formalism.
Afterwards we discuss the graphical expansion of the state  $\omega^{\beta,V}$ and its relation with the combinatoric of the imaginary time -- \textit{cf.} Theorem \ref{th:matsubara-frequency-conservation}. In section \ref{sec:graph-exp} these results will be used to obtain computation rules for the evaluation of $\omega^{\beta,V}(A)$.

\subsection{Thermal Propagators}\label{sec:propagator}

In this subsection we discuss the form of the various propagators appearing in the expansion of $\omega^{\beta,V}$. An extensive study of the analyticity properties of the various propagator presented in this section can be found in \cite{Fulling}. We start recalling the form of the two-point function for a quantum free massive scalar field propagating over a Minkowski spacetime. In particular, 
we recall that the {\bf two-point function}  of the KMS state with respect to the free time evolution at inverse temperature $\beta$ already seen in \eqref{eq:2pt-thermal-m} is
\begin{align}
\Delta^\beta_+(x) &= \frac{1}{(2\pi)^3}\int d^4p e^{ipx}  \frac{1}{1-e^{-\beta p_0}}   \delta(p^2+m^2) \text{sign}(p_0) 
\notag
\\
 &= \frac{1}{(2\pi)^3}\int d^3\mathbf{p} e^{i\mathbf{p}\mathbf{x}}  \frac{1}{1-e^{-\beta w }}  \frac{1}{2w}  \left( e^{-iw x^0}+ e^{-\beta w } e^{iw x^0}  \right)  
\label{eq:thermal-2pt}
\end{align}
where $\text{sign}(p_0)$ is the sign function and $w=\sqrt{|\mathbf{p}|^2+m^2}$. Furthermore, 
$\Delta^\beta_-(x) = \Delta^\beta_+(-x)$ and thus the {\bf commutator function} is
\[
\Delta(x) = -i(\Delta^\beta_+(x)-\Delta^\beta_-(x)).
\]
The associated {\bf Feynman propagator} is the time-ordered version of the two-point function 
\begin{align}\label{Eqn: Feynman propagator}
	\Delta^\beta_F(x) = \Theta(x^0)\Delta^\beta_+(x) + \Theta(-x^0)\Delta^\beta_-(x)
\end{align}
and the {\bf anti Feynman propagator} is the anti time-ordered two-point function
\begin{align}\label{Eqn: anti-Feymman propagator}
	\overline\Delta^\beta_F(x) = \Theta(-x^0)\Delta^\beta_+(x) + \Theta(x^0)\Delta^\beta_-(x).
\end{align}
The two-point function $\Delta^\beta_+(x)$ can be continued analytically in the imaginary times $x^0+iu$ for $u\in(-\beta,0)$: in particular we have
\[ 
\Delta_+^{\beta}(x^0+iu,\mathbf{x}) 
= \frac{1}{(2\pi)^3}\int d^3\mathbf{p} e^{i\mathbf{p}\mathbf{x}}  \frac{1}{1-e^{-\beta \omega }}  \frac{1}{2\omega}  \left( e^{u\omega}e^{ -i x^0\omega}+ e^{-\beta\omega -u\omega} e^{+ix^0\omega }  \right)\,.
\]
Following \cite{Fulling} it holds that 
$\Delta^\beta_+(z,\mathbf{x})$ can be further analitically continued to the {\bf thermal propagator}
$\Delta^\beta(z,\mathbf{x})$ which is analytic on the domain
\[
\{(z,\mathbf{x})\in\mathbb{C}\times\mathbb{R}^3|\; \text{Im}(z) \not\in \beta\mathbb{Z} \lor |\text{Re}(z)|\geq |\mathbf{x}|  \}\,.
\] 
Hence, the thermal two-point function equals the thermal propagator $\Delta^\beta(z,\mathbf{x})=\Delta_+^\beta(z, \mathbf{x})$ for $\text{Im}{(z)}\in(-\beta,0)$ and furthermore, $\Delta^\beta(z,\mathbf{x})=\Delta_-^\beta(z, \mathbf{x})$ for $\text{Im}{(z)}\in(0,\beta)$. 
We also observe that $\Delta^\beta(z,\mathbf{x})$ is periodic in the imaginary time, actually, for every $N\in\mathbb{Z}$ 
\[
\Delta^\beta(z+i N \beta,\mathbf{x}) = \Delta^\beta(z ,\mathbf{x})
\]
and, on its domain of analiticity, it is symmetric under reflection of the complex variable $z\to-z$
\[
\Delta^\beta(z,\mathbf{x}) = \Delta^\beta(-z,\mathbf{x}).
\]
Hence, the time-ordered or Feynman propagator $\Delta_F^\beta(t,\mathbf{x})$ can be obtained as the limit of $\Delta^\beta(z,\mathbf{x})$ for $\text{Im}(z)\to0$, where the limit is taken from above for $\text{Re}(z)<0$ and from below otherwise.
Because of its periodicity in the imaginary time, $\Delta^\beta(z,\mathbf{x})$ can be written as a Fourier series. 
We have in particular that  
\begin{equation}\label{eq:thermal-propagator}
	\Delta^\beta(x^0+iu,\mathbf{x}) = \frac{1}{(2\pi)^3}\int d^4p\; e^{ip x}\sum_{\eta\in\mathbb{Z}}\frac{e^{i\frac{2\pi}{\beta}\eta u}}{\omega^2 + \left(\frac{2\pi}{\beta}\eta\right)^2} (\delta(p_0-\omega)+\delta(p_0+\omega)) \left(\frac{1}{2} + i \frac{\pi}{\beta}\frac{\eta}{p_0}\right).
\end{equation}
Furthermore, the Fourier transform computed over the spatial coordinates and the real part of the time coordinate of the thermal propagator is  
\begin{equation}\label{eq:DeltaHat}
	\hat\Delta^\beta(u,p) = \frac{1}{(2\pi)^3}\sum_{\eta\in\mathbb{Z}}\frac{e^{i\frac{2\pi}{\beta}\eta u}}{\omega^2 + \left(\frac{2\pi}{\beta}\eta\right)^2} (\delta(p_0-\omega)+\delta(p_0+\omega)) \left(\frac{1}{2} + i \frac{\pi}{\beta}\frac{\eta}{p_0}\right).
\end{equation}
Notice that on $D$, the symmetries satisfied by the thermal propagator imply that $\hat\Delta^\beta(-u,-p) = \hat\Delta^\beta(u,p)$.  
Finally we observe that  the integral over $p_0$ of $\hat\Delta_+^\beta(u,p)$ defines the {\bf Matsubara propagator}
	\begin{equation}\label{eq:matsubara-propagator}
	\hat\Delta_M^\beta(u,{\mathbf{p}}) = \int dp_0 \hat\Delta_+^\beta(u,p) = \frac{1}{(2\pi)^3}\sum_{\eta\in\mathbb{Z}}  \frac{e^{i\frac{2\pi}{\beta}\eta u}}{\omega^2 + \left(\frac{2\pi}{\beta}\eta\right)^2}
	\end{equation}
where $\frac{2\pi}{\beta}\eta $ for $\eta\in\mathbb{Z}$ are known as {\bf Matsubara frequencies}.

\subsection{Graphical expansion of star and time-ordered products - real time formalism}\label{sec:graph-rt}
	
We shall now discuss the graphical expansion of $\star-$product as well as time-ordered product. This will provide a convenient graphical expansion of the Bogoliubov formula \eqref{eq:bogoliubov}.
From now on we shall denote $\star_{\Delta^\beta_+}$ with $\star$.
The graphical expansion arises from the observation that \eqref{eq:star-product-abstract} can be represented as
\begin{equation}\label{eq:star-product-2fact}
F\star G = \sum_{n\geq 0} \frac{\hbar^n}{n!} \langle F^{(n)},(\Delta_+^\beta)^{\otimes n}  G^{(n)} \rangle_n
= M e^{\Gamma^{\Delta^+_\beta}_{12}} (F\otimes G), \qquad F,G\in\mathcal{F}
\end{equation}
where $M$ is the pullback on $\mathcal{F}^{\otimes n}$ of the map $\iota:\mathcal{E}\to \mathcal{E}^{ n}$ defined by $\iota(\varphi) :=  (\varphi,\dots, \varphi)$.
Specifically we have
\begin{equation}\label{eq:M}
M(A_1\otimes \dots \otimes A_n)(\varphi)= A_1(\varphi)\dots  A_n(\varphi)\,.
\end{equation}
Furthermore, in \eqref{eq:star-product-2fact}
\begin{equation}\label{eq:Gamma-extension}
\Gamma^{\Delta_+^\beta}_{ij} = \hbar\langle \Delta_+^\beta, \frac{\delta}{\delta\varphi_i}\otimes \frac{\delta}{\delta\varphi_j}   \rangle
\end{equation}
where $\frac{\delta}{\delta\varphi_i}$ denotes the first functional derivative of the $i-$th element of a tensor product of functionals. 
In view of the associativity of this product and denoting $\Gamma_{ij}=\Gamma^{\Delta_+^\beta}_{ij}$, we have 
\begin{align}
A_1\star\dots \star A_n &= M \bigg[\bigg(\prod_{i<j}  e^{ \Gamma_{ij}}\bigg) A_1\otimes \dots \otimes A_n \bigg]
\notag
\\
 &= M \left(( e^{ \sum_{i<j} \Gamma_{ij}}) A_1\otimes \dots \otimes A_n \right)
 \notag
 \\
 &= M \circ P_n \left( A_1\otimes \dots \otimes A_n \right)
 \label{eq:star-prod}
\end{align}
where we have introduced the operator $P_n:\mathcal{F}^{\otimes n}\to \mathcal{F}^{\otimes n}$ defined as
\begin{equation}\label{eq:Pn}
P_n :=  e^{ \sum_{i<j} \Gamma_{ij}} =  \prod_{i<j} \sum_{l_{ij}=0}^{\infty} \frac{\Gamma_{ij}^{l_{ij}}}{l_{ij}!}.
\end{equation}
This operator can be seen as a formal power series in $\hbar$. The $k-$th order in this power series is a sum indexed over all possible graphs of $n$ vertices and $k$ edges. 
To formalize this observation let us denote by $\mathcal{G}_n$, the set of graphs $G$ whose vertices are $V(G)=\{1,\dots, n\}$ and with arbitrary number of edges $E(G)$.    
Hence,
\begin{equation}\label{eq:graph-exp}
P_n = \sum_{G\in\mathcal{G}_n} P_G\,,
\end{equation}
where 
\begin{equation}\label{eq:PG}
P_G := \frac{1}{\text{Sym}(G)} \left\langle p_G  , \delta_G   \right\rangle    
\end{equation}
with $\text{Sym}(G) = \prod_{i<j} l_{ij}!$ being $l_{ij}$ is the number of edges jointing the $i-$th and $j-$th vertex
of $G$.
The factors $p_G$ is defined by
\begin{equation}\label{eq:pG}
p_{G}  = \prod_{e\in E(G)} \Delta^{\beta}_{+}(x_{ e}-y_{ e})\,.
\end{equation} 
For all edge $e\in E(G)$ we shall denote by $s(e)$ (\textit{resp.} $r(e)$) the source (\textit{resp.} the target) of $e$ -- notice that $s(e)<r(e)$ for all $e\in E(G)$.
With this notation $\delta_G$ is defined by
\begin{equation}\label{eq:deltaG}
	 \delta_G =  \frac{\delta^{2|E(G)|}}{\prod_{i\in V(G)} \left(\prod_{e\in E(G): s(e)=i }\delta \varphi_i(x_{e})
	\prod_{e'\in E(G): r(e')=i } \delta \varphi_i(y_{e'}) \right)}\;.
\end{equation}
We observe that, since $\Delta_{+}^\beta$ is not symmetric, the way in which the vertices in every edge are ordered is important. 
This order is fixed by ordering the labels of the vertices of each graph, because the vertices are in one to one correspondence with the various factors $A_i$ in the $\star-$product we are analyzing.    

Since $A_i$ are not local functionals, the support of the functional derivatives $\delta_G$ in each element of the tensor product are not restricted to the diagonals. For this reason, the vertices of the graphs are not in one to one correspondence with single points of the spacetime. 

We finally stress that, even if both $p_G$ and the application of $\delta_G$ on tensor products of functionals are distributions, their composition, necessary to compute $P_G(A)$ in \eqref{eq:PG}, is always well-defined thanks to the form of the $\operatorname{WF}$ of $\Delta_+^\beta$ and to the form of the $\operatorname{WF}$ of the microcausal functionals in $\mathcal{F}$. A careful discussion of these aspects and the proof of these statements can be found in \cite{HW01} \cite{BFK}.  

\bigskip
The time-ordered product among $n$ regular functionals can be computed in a similar way as the $\star$-product. They are also expanded as a power series in $\hbar$ and the $k-$th order contribution of this series is a sum labelled over the graphs formed by $n$ vertices  and $k$ edges. The only difference with respect to the graphical expansion of the $\star$-product is the choice of the propagator used in $\Gamma_{ij}$.
Actually, 
\begin{equation}\label{eq:Tmap}
T(A_1,\dots, A_n)  =  M \bigg[\bigg(e^{ \sum_{i<j} \Gamma^{\Delta^\beta_F}_{ij}}\bigg)A_1\otimes \dots \otimes A_n\bigg]
= M \circ T_n \left( A_1\otimes \dots \otimes A_n \right)
\end{equation}
where $\Delta_F^\beta$ is the Feynman or time-ordered propagator associated to $\Delta^\beta_+$ -- \textit{cf.} equation \eqref{Eqn: Feynman propagator}. 
Hence, similarly to \eqref{eq:Pn}, 
\[
T_n:=  e^{ \sum_{i<j} \Gamma^{\Delta_F^\beta}_{ij}} =  \prod_{i<j} \sum_{l_{ij}=0}^{\infty} \frac{\Gamma_{ij}^{l_{ij}}}{l_{ij}!}
=
\sum_{G\in\mathcal{G}_n} T_G   
\]
where $T_G$ is computed in close analogy with \eqref{eq:PG}
\begin{equation}\label{eq:TG}
T_G := \frac{1}{\text{Sym}(G)} \left\langle t_G  , \delta_G   \right\rangle.   
\end{equation}
Here $\text{Sym}(G)$ and $\delta_G$ are exactly as in \eqref{eq:PG} and \eqref{eq:deltaG}.
Finally $t_G$ is constructed as in \eqref{eq:pG} where $\Delta_+^\beta$ is replaced by $\Delta_F^\beta$ 
\begin{equation}\label{eq:tG}
t_G =  \prod_{e\in E(G)} \Delta^{\beta}_{F}(x_{ e}-y_{ e}).
\end{equation}
The action of $T_G$ on a tensor product of regular functionals gives finite results.
Furthermore, even if it is not possible to extend $T_G$ over generic tensor products of microcausal functionals,   $T_G$ can be extended to tensor products of local functionals. 

As discussed in the introduction,
this extension is not unique and a renormalization procedure needs to be employed in order to give meaning to the extension of $t_G$ on the diagonals, namely over coinciding points.  
We observe that renormalization does not spoil the graphical expansion of $T_n$. It simply happens that $T_G$ in \eqref{eq:TG} is not defined with a $t_G$ as in \eqref{eq:tG} but renormalized  distributions need to be considered at the place of products of time-ordered propagators. We shall not enter this discussion here, we refer to \cite{BF00} and to \cite{HW01} for further details. 

Interacting fields are constructed by means of the Bogolibov map $R_V$ -- \textit{cf.} equation \eqref{eq:bogoliubov} -- which employs both (anti)time-ordered and $\star$- products. Therefore also $R_V$ admits an expansion in terms of graphs which involves three type of edges.
We have already seen that both the time-ordered product and the $\star$-product can be represented in graphs. Notice that, if $V$ is a local field, $S(V)$ is formally a unitary operator, and furthermore, 
\[
S^{-1}(V)  = \overline{S}(-V)
\]	
where $\overline{S}$ is the anti time-ordered exponential, namely the exponential constructed with the anti time-ordered product.
The latter is defined as in \eqref{eq:Tmap}, with the thermal anti Feynman propagator $\overline{\Delta^\beta_F}$ related to $\omega^\beta$ -- \textit{cf.} equation \eqref{Eqn: anti-Feymman propagator} -- at the place of $\Delta_F$.

\bigskip
Notably, the graphical expansion of $R_V$ in \eqref{eq:bogoliubov}, is equivalent to the graphical expansion obtained with the Feynman rules of the real time formalism of thermal field theory. Actually, in $R_V$, $\star$-products of time and anti time-ordered products appear, however these products are not mutually associative, hence the various operations need to be taken in a prescribed order.
 
In order to cope with this we double the number of field configurations, $\phi=(\varphi^1,\varphi^2)\in C^{\infty}(M)\times C^{\infty}(M)$, as it is customary in the real time formalism.
Then we use a single product rule among functionals over these field configurations, with a modified propagator which contains all the propagators described above.
This new matrix-valued propagator $D$ is defined by
\begin{equation}\label{eq:real-time-propagator}
D(x):= 
\begin{pmatrix}  \Delta_F^\beta(x) & \Delta_+^\beta(x)   \\ 
 \Delta_+^\beta(x)    & \overline{\Delta_F}^\beta(x)
\end{pmatrix}\,.
\end{equation}
In particular, if we consider regular functionals over $\Phi$ and if we denote them by $\tilde{\mathcal{F}}_{\text{reg}}$ we have that
\[
F\cdot_D G =  M e^{\tilde{\Gamma}^{D}_{12}} (F\otimes G),\qquad F,G\in\tilde{\mathcal{F}}_{\text{reg}}
\]
where now $\tilde{M}$ is the extension to $\tilde{\mathcal{F}}_{\text{reg}}^{\otimes n}$ of the map $M$ defined in \eqref{eq:M}, furthermore
$\tilde{\Gamma}^{D}_{12}$ is the extension of $\Gamma$ introduced in \eqref{eq:Gamma-extension} to $\tilde{\mathcal{F}}_{\text{reg}}^{\otimes 2}$
\[
\tilde{\Gamma}^{D}_{ij} = 
\hbar\sum_{a,b\in\{1,2\}}\langle D_{ab}, \frac{\delta}{\delta\phi^a_i}\otimes \frac{\delta}{\delta\phi^b_{j}}   \rangle
\]
where now $\frac{\delta}{\delta\phi^a_i}$ denotes the component $a$ of the first functional derivative of the $i-$th element of a tensor product of functionals.
The product defined above is a well-defined and associative on regular functionals, moreover, it admits the following graphical expansion similar to equations (\ref{eq:star-prod}-\ref{eq:Pn}-\ref{eq:graph-exp}):
\begin{align}
A_1\cdot_D \dots \cdot_D A_n  
&=  
\tilde{M} \bigg[\bigg( e^{ \sum_{i<j} \tilde{\Gamma}^{D}_{ij}}\bigg) A_1\otimes \dots \otimes A_n \bigg]
\notag
\\
&= \tilde{M} \circ 
\sum_{G\in\mathcal{G}_n} \frac{1}{\text{Sym}(G)} \left\langle d_G  , \tilde{\delta}_G   \right\rangle  
A_1\otimes \dots \otimes A_n 
,\qquad A_i\in \tilde{\mathcal{F}}_{\text{reg}}\,,
\label{eq:real-time-prop}
\end{align}
where now in analogy with (\ref{eq:tG}-\ref{eq:deltaG})
\[
d_G =  \prod_{e\in E(G)} D(x_{ e}-y_{ e})\,,\qquad
\tilde{\delta}_G =  \frac{\delta^{2|E(G)|}}{\prod_{i\in V(G)} \left(\prod_{e\in E(G): s(e)=i }\delta \phi_i(x_{e})
\prod_{e'\in E(G): r(e')=i } \delta \phi_i(y_{e'}) \right)}\,.
\]
Notice that, it is not possible to extend \eqref{eq:real-time-prop} to generic local fields because in general non local (hence non renormalizable) divergences will appear in that product. 
As an example, consider the product of $\varphi_1\varphi_2$ with $\varphi_1\varphi_1$. Then, according to \eqref{eq:real-time-prop} the Feynman propagator will be multiplied with the thermal propagator and this does not give a well-defined distribution even outside the diagonal.
This feature is also related to the fact that the time-ordered product and the $\star$-product are not mutually associative. 
In spite of this fact, we can extend \eqref{eq:real-time-prop} to tensor products of local fields of the form
\[
A_1 \otimes \dots \otimes A_n \otimes B_1\otimes \dots \otimes B_l 
\]
where $A_i$ are local fields which depends only on the second component of the field configuration, while $B_i$ are local fields which depends only on the first component of the field configuration. This extension is well-defined up to ordinary renormalization of the time-ordered products.
Notice that this extension is sufficient to construct the Bogoliubov map.

If we now denote by $\mathcal{S}(V_i)$ the exponential constructed with $\cdot_D$ and where $V_i$ is a local field which depends only on the configuration $\varphi^i$, we have the following proposition
\begin{proposition}\label{pr:real-time-formalism}
The Bogoliubov map takes the following simple form
\[
R_V(A)(\varphi) = \tilde{M} \left( \mathcal{S}(V_2) \cdot_{D} \mathcal{S}(V_1) \cdot_D {A_1}   \right)(\varphi,\varphi), \qquad A \in \mathcal{F}_{\text{loc}}\,.
\]
Furthermore, since in $(\mathcal{F}, \star_\beta)$ the state $\omega^\beta$ is nothing but the evaluation over the vanishing configuration, namely $\omega^\beta(A) = A(0)$, we have that
\begin{align}
\nonumber
\omega^\beta(R_V({A})) 
&= \sum_{n_1,n_2} \frac{i^{n_1} (-i)^{n_2}}{n_1!n_2!}  \underbrace{V_2\cdot_D \dots \cdot_D V_2}_{n_2}\cdot_D \underbrace{V_1\cdot_D \dots \cdot_D V_1}_{n_1} \cdot_D A_1\Bigg|_{(\varphi_1,\varphi_2)=(0,0)}\\
\label{Eqn: Moeller expansion}
&= \sum_{n_1,n_2} \frac{i^{n_1} (-i)^{n_2}}{n_1!n_2!}  \tilde{M}
\sum_{G\in\mathcal{G}_{n}} \frac{1}{\text{\normalfont Sym}(G)} \left\langle d_G  , \tilde{\delta}_G   \right\rangle  
\underbrace{V_2\otimes \dots \otimes V_2}_{n_2}\otimes \underbrace{V_1\otimes \dots \otimes V_1}_{n_1} \otimes A_1\Bigg|_{(\varphi_1,\varphi_2)=(0,0)}\,
\end{align} 
where $n=n_1+n_2+1$.
\end{proposition}
Notice that equation \eqref{Eqn: Moeller expansion} takes into account in a single expansion the combinatoric necessary to expand the anti-time-ordered products, the time-ordered products and the $\star$-products present in the Bogoliubov map.
Furthermore, the graphical expansion described in proposition \ref{pr:real-time-formalism} contains
the same Feynman rules of the real time formalism used in thermal field theory -- \textit{cf.} \cite{LeBellac}.

\bigskip
The expectation values discussed above depend on the time cut-off $\chi$. 
We observe that in specific cases, it is possible to obtain the equilibrium state for an interacting theory by considering a suitable limit where observables are translated at infinite positive times \cite[Thm. 3.4]{DFP-18}.
This allows to remove the time cut-off $\chi$ by translating the observables far in the future. 
In more details, considering the limit for large times and keeping $h$ fixed, we have
\begin{equation}\label{eq:return-to-equilibrium}
\lim_{t\to\infty} \omega^{\beta}(\alpha^V_t(R_V(A))) = \omega^{\beta, V}(R_V(A)), \qquad h\in C^{\infty}_0 (\Sigma)\,.
\end{equation}
In these situations we say that the return to equilibrium property holds true.
The fulfilment of \eqref{eq:return-to-equilibrium} descends from clustering properties of $\omega^\beta$ under the interacting time evolution, namely 
\begin{equation}\label{eq:clustering-condition}
\lim_{t\to\infty }\left(\omega^\beta(\alpha^V_t(A) \star B) - \omega^\beta(\alpha^V_t(A))   \omega^\beta(B)\right)  = 0, \qquad h\in C^{\infty}_0 (\Sigma),
\end{equation}
established in \cite[Prop. 3.3]{DFP-18}. 
Unfortunately, the clustering property and the return to equilibrium do not hold in the limit $h\to1$.
Hence, in that case, to evaluate the $n$-point functions on the interacting KMS state the contribution from $U(i\beta)$ needs to be taken into account.

\subsection{Graphical expansion of the state - Matsubara formalism}\label{sec:graph-matsubara}

In this section we provide a graphical expansion for the expectation value of $\omega^{\beta,V}(A)$.
Together with the result of section \ref{Sec: Expansion of the perturbative series} this would lead to practical rules for computing correlation functions for the interacting theory -- \textit{cf.} section \ref{sec:graph-exp}.
We start with formula \eqref{Eqn: cluster expansion of interacting KMS state} and we observe that the connected $n-$point functions $\omega^{\beta,c}$ admit the following graphical expansion:
\[
	\omega^{\beta,c}(A_0\otimes \dots \otimes A_n )   =  \sum_{G\in\mathcal{G}_{n+1}^c} \prod_{i<j} \left.\frac{\mathcal{D}_{ij}^{l_{ij}}}{l_{ij}!}  A_0\otimes \dots \otimes A_n \right|_{(\varphi_0,\dots ,\varphi_n)=0}\,.
\]
Here $\mathcal{G}_{n+1}^c$ is the set of connected graphs with $n+1$ vertices while $l_{ij}$ denotes the number of lines (edges) joining the $i-$th and $j-$th vertex  in the graph $G$. In addition, in the previous formula  
\[
	\mathcal{D}_{ij}:= \int dxdy\; \Delta^\beta_+(x-y) \frac{\delta}{\delta \varphi_i(x)} \otimes \frac{\delta}{\delta \varphi_j(y)}.
\]
As in \eqref{eq:PG} or in \eqref{eq:TG}, 
instead of indexing the products over the vertices of the graphs, it is easier to compute a product over the edges.
Hence, for every graph $G \in \mathcal{G}_{n+1}^c$ we have
\begin{align*}
	\omega^{\beta,c}(A_0\otimes \dots \otimes A_n )   =  \sum_{G\in\mathcal{G}_{n+1}^c}
 	\frac{1}{\text{Sym}(G)}\left. \left( \prod_{l\in E(G)}     \mathcal{D}_{s(l) r(l)}  \right) A_0\otimes \dots \otimes A_n \right|_{(\varphi_0,\dots ,\varphi_n)=0}\,,
\end{align*}
where as before $\text{Sym}(G)=\prod_{i < j} l_{ij}! $. 

We shall now prove that, because of the translation invariance of the thermal two-point function,
the integral over the simplex present in every component of $\omega^{\beta,V}$ in \eqref{Eqn: cluster expansion of interacting KMS state} can be extended to an integral over a box. To this end we need some preliminary results.

\begin{proposition}\label{pr:Fsymmetric}
For a fixed $h \in \mathcal{D}(\Sigma)$ and for $A,K\in\mathcal{A}^I(\Sigma_{2\epsilon})$, consider 
\[
F\colon\beta S_n\to\mathbb{C}\,,\qquad F(u_1,\dots, u_n) := \omega^{\beta,c}(A\otimes \alpha_{i u_1} K \otimes \dots \otimes \alpha_{i u_n} K)\,.
\]
Then $F$ can be extended to a symmetric function on  $\beta B_n\setminus D_n$ where 
$\beta B_n = (0,\beta)^{\times n}$ and $D_n=\{(u_1,\dots, u_n)\in\mathbb{R}^n| \exists i\neq j ,  u_i=u_j  \}$.
Furthermore, it holds that for $(u_1,\dots, u_n)\in \beta B_n\setminus D_n$
\begin{gather}
	F(u_1,\dots, u_n) = \sum_{G\in\mathcal{G}_{n+1}^c}\frac{1}{\text{\normalfont Sym}(G)}  F_G(u_1,\dots, u_n) := \sum_{G\in\mathcal{G}_{n+1}^c}\frac{1}{\text{\normalfont Sym}(G)} \cdot
	\notag
	\\
	\cdot\Bigg[  \prod_{l\in E(G)} \int dx_ldy_l\;   \Delta^\beta(x_{l}-y_{l} +i e_0 (u_{s(l)}-u_{r(l)}))  \frac{\delta^2}{\delta\varphi_{s(l)}(x_{l})\delta\varphi_{r(l)}(y_{l})} \Bigg] A\otimes K\otimes  \dots \otimes K \bigg|_{(\varphi_0,\dots ,\varphi_n)=0}\,,
	\label{eq:expansion-symmetry}
\end{gather}
where $\Delta^\beta$ is the thermal propagator introduced in \eqref{eq:thermal-propagator}, originally defined for $-\beta<\text{Im}(t)<0$ and extend to other values of $\text{Im}(t)\in (-\beta,\beta)\setminus \{0\}$.
\end{proposition}
\begin{proof}
Since $\dot{V}$ is a polynomial of finite order in the field, and since the lowest order in the perturbation parameter in $K$ is $1$, at fixed perturbation order only finitely many connected graphs contribute to the sum defining  $F(u_1,\dots, u_n)$. 
Consider a connected graph $G\in\mathcal{G}_{n+1}^c$. The corresponding contribution to $F$ for $(u_1,\dots, u_n)\in \beta S_n$ is 
\begin{equation}\label{eq:step1}
F_G(u_1,\dots, u_n) =  \int dXdY \left( \prod_{l\in E(G)} \;   \Delta^\beta_+(x_{l}-y_{l} +ie_0 (u_{s(l)}-u_{r(l)}))\right) \Psi_G(X,Y)
\end{equation}
where $X=(x_1,\dots x_{|E(G)|})$, $Y=(y_1,\dots y_{|E(G)|})$ and with the caveat that $u_0=0$. Furthermore, 
\begin{equation}\label{eq:psiG}
\Psi_G(X,Y) :=  \left.\left( \prod_{l\in E(G)} \frac{\delta^2}{\delta\varphi_{s(l)}(x_{l})\delta\varphi_{r(l)}(y_{l})} \right)  A\otimes K\otimes  \dots \otimes K \right|_{(\varphi_0,\dots ,\varphi_n)=0}.
\end{equation}
Notice that $\Psi_G$ is a compactly supported distribution and hence its Fourier transform $\hat{\Psi}_G(-P,P)$ is an analytic function which grows at most polinomially in some directions -- here $P=(p_1,\dots, p_{|E(G)|})$.

In view of the properties of the thermal two-point function and since $\Delta^\beta_+(x+ie_0s) = \Delta^\beta(x+ie_0s)$ for $s\in (-\beta,0)$, we have that for $U=(u_1,\dots, u_n)\in \beta S_n$
\begin{equation}\label{eq:step2}
F_G(U) =  \int dXdY \left( \prod_{l\in E(G)} \;   \Delta^\beta(x_{l}-y_{l} +ie^0 (u_{s(l)}-u_{r(l)}))\right) \Psi_G(X,Y)\,.
\end{equation}
Furthermore, the analyticity properties of the thermal two-point function imply that $F_G$ can be extended to $U\in \beta B_n \setminus D_n$.
We have thus proved that $F$ can be extended to $\beta B_n \setminus D_n$ and that equation \eqref{eq:expansion-symmetry} holds true.

To conclude the proof we have to check that $F$ is symmetric.
Let $U\in \beta B_n\setminus D_n$: Then there exists a permutation $p$ which orders the entries of $U$ and such that $p(U)\in \beta S_n$.
We shall prove that $F(U)=F(p(U))$.
For every graph $G\in\mathcal{G}^c_{n+1}$, we denote with $p(G)$ the (unique) graph $p(G)\in\mathcal{G}^c_{n+1}$ (possibly equal to $G$) obtained applying the permutation $p$ on the last $n$ vertices of $G$. We now prove that
\[
F_G(U) = F_{p(G)}(p(U)).
\]
To prove the previous formula we recall that $\Delta^\beta (x+ise_0) = \Delta^\beta (-x-ise_0)$.
Furthermore, the source and the range of the lines in \eqref{eq:step2} are ordered in such a way that $s(l)<r(l)$.
For every $l\in E(G)$, we consider the corresponding edge $p(l)\in p(G)$.
There are two possibilities, either $p(s(l))=s(p(l))$ or  $p(s(l))=r(p(l))$. In the first case $u_{s(l)} = p(u)_{s(p(l))}$, $u_{r(l)} = p(u)_{r(p(l))}$, $x_{p(l)}=x_{l}$ and $y_{p(l)}=y_{l}$ so,
\[
\Delta^\beta(x_{p(l)}-y_{p(l)} +i e_0 (p(u)_{s(p(l))}-p(u)_{r(p(l))}))
=\Delta^\beta(x_{l}-y_{l} +i e_0(u_{s(l)}-u_{r(l)}))\,.
\]
If, on the other hand, $p(s(l))=r(p(l))$ we have that $u_{s(l)}=p(u)_{r(p(l))}$, $u_{r(l)}=p(u)_{s(p(l))}$, $x_{p(l)}=y_{l}$ and $y_{p(l)}=x_{l}$.
Hence,
\begin{align*}
\Delta^\beta(x_{p(l)}-y_{p(l)} +i e^0(p(u)_{s(p(l))}-p(u)_{r(p(l))})) 
&=
\Delta^\beta(y_{l}-x_{l} +i e_0(u_{r(l)}-u_{s(l)}))
\\
&=
\Delta^\beta(x_{l}-y_{l} +i e_0(u_{s(l)}-u_{r(l)}))\,,
\end{align*}
because $\Delta^\beta (x+ise_0) = \Delta^\beta (-x-ise_0)$. Repeating the analysis for every edge of $G$ we prove the desired symmetry.
\end{proof}
		
The following theorem uses the symmetrisation  of $F_n$ established in Proposition \ref{pr:Fsymmetric} to rewrite the integral over the simplex of $F$ as an integral over a box. The thermal propagators present in the resulting function can be expanded as a sum over the Matsubara frequencies as in \eqref{eq:thermal-propagator}. Furthermore, the integral over the box can be taken and it furnishes the conservation of the Matsubara frequency at every vertex.   		
		
\begin{theorem}\label{th:matsubara-frequency-conservation}
For a fixed $h \in \mathcal{D}(\Sigma)$, and for $A,K\in\mathcal{A}^I(\Sigma)$, as in Proposition \ref{pr:Fsymmetric}, consider 
\[
F(u_1,\dots, u_n) := \omega^{\beta,c}(A\otimes \alpha_{i u_1} K \otimes \dots \otimes \alpha_{i u_n} K).
\]
Let $\beta S_n=\{(u_1,\dots, u_n)| 0<u_1<\dots < u_n <\beta\}$ and $\beta B_n=(0,\beta)^{\times n}$.
It holds that 
\begin{align*}
\int_{\beta S_n} F(u_1,\dots, u_n) du_1\dots du_n  
&=  \frac{1}{n!}\int_{\beta B_n} F(u_1,\dots, u_n) du_1\dots du_n \\
&=  \sum_{G\in\mathcal{G}_{n+1}^c}\frac{1}{\text{\normalfont Sym}(G)}\frac{1}{n!}\int_{\beta B_n} F_G(u_1,\dots, u_n) du_1\dots du_n
\end{align*}
where $F_G$ is as in \eqref{eq:expansion-symmetry}. 
Furthermore, the integral over $\beta B_n$ gives the conservation of the Matsubara frequency over every vertex of the graph, namely,
\begin{gather*}
\int_0^\beta du_1\dots \int_0^\beta du_n F_G(u_1,\dots, u_n) = \frac{1}{\beta^{|E(G)|}} 
	\int dP \sum_{N\in{\mathbb Z}^{|E(G)|}} \left(\prod_{l\in E(G)} \;   
	\tilde\Delta^\beta(n_l,p_l)\right) 
	\\
	\left(\prod_{1\leq j\leq n} \delta\left(\sum_{l'\in E(G), r(l')=j} n_{l'}  -\sum_{{l''}\in E(G), s(l'')=j} n_{l''}\right)\right)
	\hat{\Psi}_G(-P,P)
\end{gather*}
where $N=(n_1,\dots, n_{|E(G)|})$, $\delta$ is the Kronecker delta function, $\tilde{\Delta}^\beta(n_l,p_l)$ is the expansion in Matsubara frequencies of 
\eqref{eq:DeltaHat}, namely
\[
\tilde\Delta^\beta(n,p) =  \frac{1}{\omega^2 + \left(\frac{2\pi}{\beta}n\right)^2} (\delta(p_0-\omega)+\delta(p_0+\omega)) \left(\frac{1}{2} + i \frac{\pi}{\beta}\frac{n}{p_0}\right)
\]
where $\omega=\sqrt{\mathbf{p}^2+m^2}$.
Furthermore, $\hat{\Psi}$ is the Fourier transform of the function $\Psi$ given in \eqref{eq:psiG}.
\end{theorem}
\begin{proof}
Let $\beta D_n=\{(u_1,\dots, u_n)\in \beta B_n| u_i=u_j \text{ for some } i\neq j \}$ be the set of diagonals in $\beta B_n$.
$\beta B_n\setminus \beta D_n \subset \beta B_n$ is then formed by $n!$ disjoint components. 
As discussed in Proposition \ref{pr:Fsymmetric}, the function $F(u_1,\dots, u_n)$ originally defined on $\beta S_n$ can be extended to a symmetric function on $\beta B_n$ up to its diagonals. 
The extended $F$ is symmetric, namely $F(p(U))=F(U)$ for every permutation $p$ of the components of $U=(u_1,\dots, u_n)$.
Hence, the integral over $\beta B_n$ is equivalent to the integral over the simplex $\beta S_n$ up to the number of possible permutations which is $n!$. This proves the first equality of the theorem.

\bigskip \noindent
To prove the conservation of the Matsubara frequencies, namely the second equality of the theorem, we start observing that, 
for every $G\in\mathcal{G}_{n+1}^c$, the corresponding  $F'_G$ can be written with the help of the Fourier transform.

In view of the translation invariance of $\Delta^\beta$ and considering the Fourier transform with respect to  $X$ and $Y$, \eqref{eq:step1} implies that \begin{equation}\label{eq:step2-symmetry}
F_G(u_1,\dots, u_n) =  \int dP  \left(\prod_{l\in E(G)} \;   \hat\Delta^\beta(u_{s(l)}-u_{r(l)},p_l)\right) \hat{\Psi}_G(-P,P)
\end{equation}
where $\hat\Delta^\beta(u,p)$ is given in \eqref{eq:DeltaHat} and is equal to $\hat\Delta^\beta_+(u,p)$ for $u\in (-\beta,0)$. 

\bigskip
Moreover we have
\[
\int_{\beta B_n} dU F_G(u_1,\dots, u_n) =  \int dP \int_{\beta B_n} dU   \left(\prod_{l\in E(G)} \;   \hat\Delta^\beta(u_{s(l)}-u_{r(l)},p_l)\right) \hat{\Psi}_G(-P,P)\,,
\]
namely it is possible to consider the $U$-integral before the $P$-one.
To prove this we first of all mollify every delta functions in the propagators $\hat{\Delta}_+^\beta$.
Then we divide the domain of $U-$integration in $n!$ disjoint contributions obtained removing all the diagonal $\beta B
_n \setminus \beta D_n$.
Up to a permutation of the $U$'s entries $G'=\pi(G)$, every disjoint component of $\beta B
_n \setminus \beta D_n$ can be seen as a an integral over $\beta \mathcal{S}_n$.
We consider one of these components and let $\epsilon\in(0,\beta /(n+1))$. 
We divide the domain of integration in two disjoint parts 
\begin{align*}
	\beta S_n = \beta S_n^- \cup \beta S_n^\epsilon\,,\qquad
	\beta S_n^\epsilon = \{U\in \beta S_n | \beta-u_n <\epsilon \}\,,\qquad
	\beta S_n^-=\beta S_n \setminus \beta S_n^\epsilon\,.
\end{align*}
Then let
\[
g_G(u_0,u_1,\dots,u_n;p_{1},\dots,p_{E(G)} ):=\prod_{l\in E(G)} \;   \hat\Delta^\beta(u_{s(l)}-u_{r(l)},p_l)
\]
and
\[
A_G =\int_{\beta S_n^-} dU   \int dP g_G (0,u_1,\dots,u_n;p_{1},\dots,p_{E(G)})  \hat\Psi_G(-P,P).
\]
In $A_G$ we can exchange the $U$-integral with the $P$-integral
To prove this we observe that $\hat\Psi_G(-P,P)$ is the Fourier transform of a compactly supported distribution, hence it is  an analytic function which grows at most polynomially for large $|P|$. Notice however that  
$g_G(u_0,U,P) \Psi_G(-P,P)$ is of rapid decrease in every direction uniformly in $U$ for $U\in\beta S_n^-$.
Actually, the presence of the mollified delta functions in the propagators the only directions which can be of non rapid decrease for $g_G(u_0,U,P) \Psi_G(-P,P)$ 
are those in which every $p_l\in P$ is null. Furthermore, as shown in Proposition 9 of \cite{FredenhagenLindner},
 since $A$ and $K$ are microcausal functionals, the directions where every component of $P$ is in the forward light cone are of rapid decrease for $\hat\Psi_G(-P,P)$. 
Finally, if some past directed $p_l$ are considered, we notice that 
\[
|g_G(0,U,P)| \leq C e^{-(\beta-M)\omega(\mathbf{p}_l)}, \qquad U\in \beta S_n^-
\]
where $M=\sup|u_i-u_j|<(\beta -\epsilon)$.  

The contribution over $\beta S_n^\epsilon$ can be treated in a similar way. 
Actually, in view of the KMS property the domain of the $U$-integral can be conveniently modified as follows.
The KMS property ensures that
\[
g_G (u_0,u_1,\dots,u_n;p_{1},\dots,p_{E(G)} ) = g_{G'}(u_n-\beta,u_0,u_1\dots,u_{n-1};q_{1},\dots,q_{E(G')} )\,,
\]
where $G'=\pi(G)$, being $\pi$ the cyclic permutation of the vertices which maps the last vertex into the first one, namely $\pi(G)=\{n,0,1,2,\dots, n-1\}$; Furthermore $q_\pi(l) = -p_l$ if $n\in l$ and $q_\pi(l) = p_l$ otherwise. 
Using the invariance under rigid translation of the $u$ variable, we perform the following change of coordinates:
\[
(v_0,\dots, v_n) := (u_n-\beta,u_0,u_1\dots,u_{n-1})+(\beta-u_n,\dots, \beta-u_n )\,.
\] 
We observe that $U=(u_1,\dots, u_n)\in \beta S_n^\epsilon$ if and only if  $V=(v_1,\dots v_n) \in \beta S_{n}^{(1)} = \{ V\in \beta S_n| v_1<\epsilon \}$. 
We thus have
\begin{gather*}
\int_{\beta S_n^\epsilon} dU  \int dP  g_G(u_0,u_1,\dots,u_n;p_{1},\dots,p_{E(G)} ) \Psi_G(-P,P) = 
\\  \int_{\beta S_{n}^{(1)}} dV  \int dQ g_{G'}(v_0,\dots v_n;q_{1},\dots,q_{E(G')} ) \Psi_{G'}(-Q,Q).
\end{gather*}
We can now split the integral over $\beta S_{n}^{(1)}$ in two parts, $\beta S_{n}^{(1)}=(\beta S_{n}^{(1)})^-\cup (\beta S_{n}^{(1)})^\epsilon$ in a similar way as we have done before. Again, the contribution due to $(\beta S_{n}^{(1)})^-$ can be treated as $A_G$ above while we can use the KMS condition to modify the domain of integration of $(\beta S_{n}^{(1)})^\epsilon$ to $\beta S_{n}^{(2)}:=\lbrace V\in \beta S_n| v_1<\epsilon\,,v_2-v_1<2\epsilon\rbrace$. 
If we repeat the procedure $n$ times, we end up with a domain of integration $\beta S_{n}^{(n)}$ where the largest distance between the arguments is $ n \epsilon$. This last contribution can be again treated as $A_G$, namely, its integrand is of rapid decrease in every direction. 
This implies that we can take the integral over $U$ before the integral over $P$.

Consider now the integral of $F_G$ over $\beta B_n$ and expanding the propagator $\Delta^\beta$ as sums over the Matsubara frequencies according to \eqref{eq:DeltaHat} we have
\[
I_G=
\int dP \int_{\beta B_n} dU   \left(\prod_{l\in E(G)} \;   
\sum_{n_l=-\infty}^\infty
e^{i\frac{2\pi}{\beta}n_l (u_{s(l)}-u_{r(l)})}
\tilde\Delta^\beta(n_l,p_l) 
\right) \hat{\Psi}_G(-P,P)
\]
where $\omega_l =\sqrt{\mathbf{p}_l^2+m^2}$.
Reordering the sum and the products in order to collect the exponentials over the vertices we have
\begin{gather*}
I_G=
\int dP \int_{\beta B_n} dU   
\sum_{N} 
\exp{ \left(  i\frac{2\pi}{\beta} \sum_{j=1}^n \left(\sum_{l'\in E(G), s(l')=j} n_{l'}  -\sum_{{l''}\in E(G), r(l'')=j} n_{l''}\right)      u_j \right) }
\\
\cdot \left(\prod_{l\in E(G)} \;
 \tilde\Delta^\beta(n_l,p_l) 
 \right) \hat{\Psi}_G(-P,P)\,.
\end{gather*}
In view of the form of $\tilde\Delta^\beta(n_l,p_l)$ we notice that the integral over all $U$ can now be taken before the sum over $N$. We observe that the integral vanishes unless there is conservation of the Matsubara frequencies on every vertex of the graph $G$. This proves the last equality of the theorem.
\end{proof}

We shall now discuss the limit $h\to1$ of
\begin{align*}
	L_h=\int dU \omega^{\beta,c}(A\otimes \alpha_{iu_1} K \otimes\dots \otimes \alpha_{iu_n} K)\,.
\end{align*}
To this end we write 
\begin{align*}
	K = \int_{\mathbb{R}^3} d^3\mathbf{x} \mathcal{H}_h(\mathbf{x})\,,\qquad
	\mathcal{H}_h(\mathbf{x})   = \int dt\;\dot{\chi}(t) R_{V_h} \left(\mathcal{L}_I(t,\mathbf{x})\right)\,.
\end{align*}
Furthermore, in the limit $h\to1$ we have
\[
\mathcal{H}(\mathbf{x}) = \alpha^{V}_{0,\mathbf{x}}\mathcal{H}(0)\,,
\]
where $\alpha^V_{t,\mathbf{x}}$ denotes the spacetime (interacting) translation of the step $(t,\mathbf{x})$
\[
\alpha^V_{t,\mathbf{x}} R_V \mathcal{L}_I(s,\mathbf{y}) = R_V \mathcal{L}_I(s+t,\mathbf{y}+\mathbf{x}).
\]
Notice that $\mathcal{H}_h(\mathbf{x})$ depends on $h$ through $V_h$ in the Bogoliubov map.
However, thanks to the causal factorisation property and the support property of $\chi$, $\mathcal{H}(\mathbf{x})=\lim_{h\to1} \mathcal{H}_h$ is of compact support and thus it is an element of $\mathcal{F}$ which does not depend on $h$ anymore.

Hence, we consider 
\begin{equation}\label{eq:Lh}
L_h=\int d\mathbf{x}_1\dots d\mathbf{x}_n h(\mathbf{x}_1)\dots h(\mathbf{x}_n) \int dU \omega^{\beta,c}(A\otimes \alpha_{iu_1}\alpha^V_{0,\mathbf{x}_1} \mathcal{H}(0) \otimes\dots \otimes \alpha_{iu_n}\alpha^V_{0,\mathbf{x}_n} \mathcal{H}(0)).
\end{equation}
We are eventually interested in discussing the limit $h\to1$ of the previous expression. 
Notice that, in this way, we are actually taking the adiabatic limit $\lim_{h\to1}\omega^{\beta,V}(R_V(A))$ in two steps.
We are considering first of all the limit $h\to 1$ of the potential in the Bogoliubov map and then the limit $h\to1$ of every $L_h$. Further details on this steps can be found in \cite{FredenhagenLindner} and in the appendix of \cite{DHP}.

\begin{theorem}\label{th:momentum-conservation}
The limit $h\to1$ of $L_h$ is well-defined.
Furthermore, considering the expansion in terms of connected graphs of $L_h$
\[
L_h=\sum_{G\in\mathcal{G}_{n+1}^c}  \frac{1}{ {\text{\normalfont Sym}(G)}  n!}    L_{h,G},
\]
it holds that, for every $G$ in $\mathcal{G}_{n+1}^c$, the limit $h\to 1$ results in the spatial momentum conservation in every vertex of $G$.
\end{theorem}
\begin{proof}
As discussed above, in the limit $h\to 1$,
$\mathcal{H}(\mathbf{x})$ is of compact support and it is an element of $\mathcal{F}$. Furthermore, $\alpha^V_{0,\mathbf{x}}\mathcal{H}(\mathbf{y})$ converges to $\alpha_{0,\mathbf{x}}\mathcal{H}(\mathbf{y})$ as $h\to 1$. Therefore we can expand $L_{h,G}$ as follows
\[
L_{h,G}=\int d\mathbf{x}_1\dots d\mathbf{x}_n h(\mathbf{x}_1)\dots h(\mathbf{x}_n) W(\mathbf{x}_1,\dots,\mathbf{x}_n)
\]
where
\begin{gather*}
W(\mathbf{x}_1,\dots,\mathbf{x}_n) = \int_{\beta B_n} dU \int dP   
\left(\prod_{e=1}^n  
\exp{\left(
i\mathbf{x}_e \cdot 
\left(
\sum_{l'\in E(G), s(l')=e} \mathbf{p}_{l'}  
-\sum_{{l''}\in E(G), r(l'')=e} \mathbf{p}_{l''}
\right) 
\right) } 
\right)\;
\\
\cdot\left(\prod_{l\in E(G)} \;   \hat\Delta^\beta(u_{s(l)}-u_{r(l)},p_l)\right)
\hat{\Phi}_G(-P,P)
\end{gather*}
and now $\hat\Phi_G$ is the Fourier transform of the compactly supported distribution $\Phi_G$ whose integral Kernel is of the form
\[
\Phi_G(X,Y) :=  \left.\left( \prod_{l\in E(G)} \frac{\delta^2}{\delta\varphi_{s(l)}(x_{l})\delta\varphi_{r(l)}(y_{l})} \right)  A\otimes \mathcal{H}(0)\otimes  \dots \otimes \mathcal{H}(0) \right|_{(\varphi_0,\dots ,\varphi_n)=0}.
\]
The function $W$ is an integrable function over $\mathbb{R}^{3n}$.  To prove it we start mollifying the delta functions in the propagators -- \textit{cf.} equation \eqref{eq:thermal-2pt}.
Since the mollified delta functions $\delta_\epsilon$ converge to the Dirac distribution $\delta$ in the distributional sense, we can consider the limit $\delta_\epsilon\to\delta$ in the last step. 
Since $A$ and $\mathcal{H}(0)$ are of compact support, we can prove as in the proof of Theorem \ref{th:matsubara-frequency-conservation} that  $\Phi(-P,P) \prod_l \hat\Delta^\beta(u_{s(l)}-u_{r(l)},p_l) $ is uniformly bounded in $U$ by a rapidly decreasing function. This implies that the integral in $U$ (which is done over a compact support) can be taken before the integral in $P$.  

\bigskip \noindent
The resulting function 
\[
V(P) := \int_{\beta B_n} dU \left( \prod_{l\in E(G)} \hat\Delta^\beta(u_{s(l)}-u_{r(l)},p_l) \right)\Phi(-P,P)
\] is rapidly decreasing in every direction.
This implies that the limit $h\to1$ in $L_{h,G}=\langle W,h\otimes\dots \otimes  h\rangle$ can be taken before the integral over $\mathbf{X}=(\mathbf{x}_1,\dots ,\mathbf{x}_n)$. 
\noindent
We finally need to evaluate the following oscillatory integral obtained as $L_G=\lim_{h\to 1} L_{h,G}$
\begin{align*}
L_G&=\int_{\mathbb{R}^{3n}} d\mathbf{X} W(\mathbf{X})  \\
&= \int_{\mathbb{R}^{3n}} d\mathbf{X} \int dP \left(\prod_{e=1}^n  
\exp{\left(
i\mathbf{x}_e \cdot 
\left(
\sum_{l'\in E(G), s(l')=e} \mathbf{p}_{l'}  
-\sum_{{l''}\in E(G), r(l'')=e} \mathbf{p}_{l''}
\right)
\right) } 
\right)\; V(P).  
\end{align*}
Since $V(P)$ is of rapid decrease we have that 
\[
L_G= (2\pi)^{3n}\int dP \left(\prod_{e=1}^n  
\delta{
\left(
\sum_{l'\in E(G), s(l')=e} \mathbf{p}_{l'}  
-\sum_{{l''}\in E(G), r(l'')=e} \mathbf{p}_{l''}
\right) } 
\right)\; V(P).
\]
The final step is to recall that, as in Theorem \ref{th:matsubara-frequency-conservation},  $V(P)$ can be expanded as a sum over Matsubara frequencies, hence, 
\begin{gather*}
L_{G} = \frac{(2\pi)^{3n}}{\beta^{|E(G)|}}\int dP   
\left(\prod_{e=1}^n  
\delta { 
\left(
\sum_{l'\in E(G), s(l')=e} \mathbf{p}_{l'}  
-\sum_{{l''}\in E(G), r(l'')=e} \mathbf{p}_{l''}
\right) } 
\right)\;
\\
\cdot\sum_{N} 
{ \left( \prod_{j=1}^n  \delta\left(\sum_{l'\in E(G), s(l')=j} n_{l'}  -\sum_{{l''}\in E(G), r(l'')=j} n_{l''}\right)      \right) }
\cdot \left(\prod_{l\in E(G)} \;
 \tilde\Delta^\beta(n_l,p_l) 
 \right) \hat{\Phi}_G(-P,P)
\end{gather*}
If follows that the limit $h\to 1$ ensures momentum conservation in every vertex of the graph $G$.
Actually, the limit vanishes unless the sum of incoming / outgoing momenta vanishes.
\end{proof}

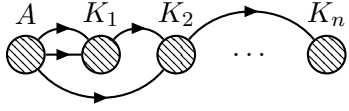
\begin{figure}
\centering
		\begin{tikzpicture}[baseline={([yshift=-.5ex]current bounding box.center)},thick]
			\begin{feynman}[inline=(d.base)]
				\vertex [blob] (A) {};
				\vertex [right=of A, blob] (k1) {};
				\vertex [right=of k1, blob] (k2) {};
				\vertex [right=of k2] (kd) {\dots};
				\vertex [right=of kd, blob] (kn) {};
				\vertex [above=1.4em of k1] {\({K_1}\)};
				\vertex [above=1.4em of A] {\(A\)};
				\vertex [above=1.4em of k2] {\(K_2\)};
				\vertex [above=1.4em of kn] {\(K_n\)};
								
				\diagram* { 
					(A) -- [fermion, bend left=50] (k1);
					(A) -- [fermion] (k1);
					(A) -- [fermion, bend right=50] (k2);
					(k1) -- [fermion, bend left=50] (k2)
					(k2) -- [fermion, bend left=50] (kn)
				};
			\end{feynman} 
		\end{tikzpicture} 
\caption{This figure contains a graph appearing in the expansion of $\omega^{\beta,V}(A)$. $K$ correspond to the interacting Hamiltonian and the  arrows correspond to the propagators given in \eqref{eq:thermal-propagator}. According to Theorem \ref{th:momentum-conservation} and Theorem \ref{th:matsubara-frequency-conservation}, at each vertex there is the conservation of the spatial momentum and the Matsubara frequencies.}
\label{Fig: graph appearing in the expansion of correlation functions}
\end{figure}

We conclude this section by comparing the obtained graphical expansion with the Matsubara formalism.
First of all we notice that in the limit where $\chi(t)$ tends to an Heaviside step function $\Theta(t)$, the support of $\mathcal{H}(0)$ tends to be the point $(0,0)$.
This limit is in general singular, however, we observe that, at least formally, in this situation the propagators $\Delta^\beta(t+iu)$  joining the various $\mathcal{H}(0)$ tend to the ordinary Matsubara propagators \eqref{eq:matsubara-propagator}. 
The time dependence survives only in the external lines and if we evaluate also those fields at $t=0$ we obtain the following (formal) graphical expansion:
\[
\int_{\beta S_n} dU \omega^{\beta,c}(A\otimes \alpha_{iu_1}K\otimes\dots\otimes \alpha_{iu_n}K)\,.
\]
Therefore, at least formally, the expansion found in proposition \ref{pr:Fsymmetric} and theorem \ref{th:matsubara-frequency-conservation} coincides with the usual Matsubara expansion -- \textit{cf.} \cite{LeBellac} -- in the limit where $h\to1$ and $\chi \to \theta$.

\subsection{Computational rules - Graphical expansion of the correlation functions}\label{sec:graph-exp}

In this section we summarise the computational rules necessary to evaluate the $n-$th order contribution to the correlation functions 
\[
G(x_1,\dots, x_k)=\omega^{\beta,V}(R_V(T(\Phi(x_1),\dots ,\Phi(x_k)))\,,
\]
in the state $\omega^{\beta,V}$  \eqref{Eqn: cluster expansion of interacting KMS state} where the interaction Lagrangian is $\mathcal{L}_I=\lambda\phi^l$.  These rules follows from the analysis performed in the previous section and replace both the ordinary Feynman rules for the thermal field theory and the rules used in the  ordinary Matsubara computations. 
In particular, we shall use an expansion like the one given in \eqref{eq:Lh} where the interaction Hamiltonians $K$ are replaced by their spatial densities $\mathcal{H}(0)$.
Furthermore proposition \ref{pr:real-time-formalism} is used to evaluate the action of the Bogolibov map $R_V$, while Theorem  \ref{th:matsubara-frequency-conservation} allows to compute the integral over the imaginary times.
Finally Theorem \ref{th:momentum-conservation} gives the conservation of the spatial momentum over each vertex representing the interaction Hamiltonian.

The contribution $\hat{G}^{[n]}(t_1,\mathbf{p_1};\dots ; t_k,\mathbf{p_k})$ at order $n$ in perturbation theory to the spatial Fourier transform
of $G(x_1,\dots, x_k)$ is given by the sum of all contributions obtained with the following rules valid when $\mathcal{L}_I$ is a of the form  $\lambda \phi^l$ -- \textit{cf.} figure \ref{Fig: graph appearing in the expansion of correlation functions}.

\begin{enumerate}
\item  Add internal vertices for $A=R_V(T\Phi(t_1,\mathbf{p}_1), \dots ,\Phi(t_1,\mathbf{p}_k))$ and for each factor $\mathcal{H}(0)= R_V\left(\int dt \dot{\chi}(t)\mathcal{L}_I(t,0)\right)$. 
The number of $\mathcal{H}(0)$ factors plus the number of internal vertices must be equal the order $n$ of perturbation.
Moreover, there are two types of internal vertices corresponding to the two-components of the propagator \eqref{eq:real-time-propagator}.
Notice that the external vertices and those corresponding to $\mathcal{L}_I$ are always of type $1$.
\item  Join the vertices in $A$ and in every $\mathcal{H}$ with real time propagators given in \eqref{eq:real-time-propagator}. Only connected graphs are allowed.
\item  Use the thermal propagator expanded over the Matsubara frequencies \eqref{eq:thermal-propagator} with $u=0$  to join the vertices of $A$ with the one of the factors  $\mathcal{H}$, in such a way that at every internal vertex the number of lines correspond to the order of $\mathcal{L}_I=\lambda\phi^l$. 
Again only connected graphs among $A$ and the various $\mathcal{H}(0)$ are allowed.
\item  Impose spatial momentum conservation in every internal vertex and the vanishing of the sum of Matsubara frequencies over the factors $A$, $\mathcal{H}(0)$.
\item  Add the appropriate numerical factors to every graph.
\item  Perform the integrations over times (taking into account $\chi$, $\dot{\chi}$).
\end{enumerate}
The last step in the list is the most complicated one because it involves integrations over times and, furthermore, it formally depends on the form of the $\dot{\chi}$. 
It is however important to notice that although $\chi$ appears in this analysis we have that the final result does not depend on $\chi$ if all the points $x_i$ are in the region where $\chi=1$. This observation can be used to simplify some of these integrals.

In particular, we observe that, in the limit $j\to1$, $\chi_j(t) = \chi(t/j)$ tends to $1$ while ${\dot\chi}_{j}\to 0$ uniformly and in the Fourier domain $\hat{\dot\chi}_{j}$ tends to $0$ pointwise. 
As we shall see in section \ref{sec:practical}, with this observation many contributions cancel in the final expressions.

Finally, it is interesting to notice that, in order to evaluate
\begin{align*}
	G(x_1,\ldots,x_n)=\omega^{\beta,V}\left(R_V(T(\phi(x_1),\dots ,\phi(x_n)))\right)\,,
\end{align*}
the rules presented above combine both the Feynman rules proper of the real time formalism and the Matsubara rules.
In particular real time formalism is needed to evaluate the effect of $R_V(T(\phi(x_1),\dots ,\phi(x_n)))$, while the Matsubara formalism ensures that $U(i\beta)$ is correctly taken into account.

If the interaction Lagrangian in $(\mathcal{F},\star_{\Delta_+^\beta})$ results in a polynomial in $\phi$,
we must add as many types of internal vertices as the number of homogeneous monomials forming $\mathcal{L}_I$. 
Notice that this is not an uncommon situation: Actually, in a $\lambda \phi^4-$theory, in view of the form of \eqref{eq:map-monomials}, the interaction Lagrangian when represented over $(\mathcal{F},\star_{\Delta_+^\beta})$ is such that 
\begin{equation}\label{eq:phi4-rep}
r_{\Delta_+^\beta}\left(\frac{\Phi^4}{4!}\right) = \frac{\phi^4}{4!} + m_\beta^2 \frac{\phi^2}{2} 
\end{equation} 
where $m_\beta$ is the known thermal mass.
As an example, figure \ref{fig:diagrams-phi4} contains all the diagrams which contribute to $G(x_1,x_2)=\omega^{\beta,V}(R_V(T(\phi(x_1),\phi(x_2)))$ at second order in $\lambda$ where one vertex is $\phi^4$ and the other  is $m_\beta^2 \phi^2$. 

\begin{figure} 
\centering
\begin{tikzpicture}
	\begin{feynman}[inline=(d.base)]
				\vertex (x);
				\vertex [right=of x] (z1);
				\vertex [above=of z1] (z2);
				\vertex [right=of z1] (y);
				\vertex [above =.1em of z2] (y1) {\(\chi\)};
				\vertex [below =.1em of z1] (y2) {\(\chi\)};
				\diagram { 
					(x) -- (z1);
					(z1) -- [quarter left] (z2);
					(z1) -- [quarter right] (z2);
					(z1) --  (y);
				};
			\end{feynman}			
\end{tikzpicture}
\quad
\begin{tikzpicture}
	\begin{feynman}[inline=(d.base)]
				\vertex (x);
				\vertex [right=of x] (z1);
				\vertex [above=of z1] (z2);
				\vertex [right=of z1] (y);
				\vertex [above =.1em of z2] (y2) {\(\dot\chi\)};
				\vertex [below =.1em of z1] (y1) {\(\chi\)};
				\diagram { 
					(x) -- (z1);
					(z1) -- [scalar,quarter left] (z2);
					(z1) -- [scalar,quarter right] (z2);
					(z1) --  (y);
				};
			\end{feynman}			
\end{tikzpicture}
\quad
\begin{tikzpicture}
	\begin{feynman}[inline=(d.base)]
				\vertex (x);
				\vertex [right=of x] (z1);
				\vertex [above=of z1] (z2);
				\vertex [right=of z1] (y);
				\vertex [above =.1em of z2] (y2) {\(\chi\)};
				\vertex [below =.1em of z1] (y1) {\(\dot\chi\)};
				\diagram { 
					(x) -- [scalar] (z1);
					(z1) -- [quarter left] (z2);
					(z1) -- [quarter right] (z2);
					(z1) -- [scalar] (y);
				};
			\end{feynman}			
\end{tikzpicture}
\quad
\begin{tikzpicture}
	\begin{feynman}[inline=(d.base)]
				\vertex (x);
				\vertex [right=of x] (z1);
				\vertex [above=of z1] (z2);
				\vertex [right=of z1] (y);
				\vertex [above =.1em of z2] (y2) {\(\dot\chi\)};
				\vertex [below =.1em of z1] (y1) {\(\chi\)};
				\diagram { 
					(x) -- [scalar](z1);
					(z1) -- [quarter left] (z2);
					(z1) -- [quarter right] (z2);
					(z1) -- [scalar] (y);
				};
			\end{feynman}			
\end{tikzpicture}
\quad
\begin{tikzpicture}
	\begin{feynman}[inline=(d.base)]
				\vertex (x);
				\vertex [right=of x] (z1);
				\vertex [above=of z1] (z2);
				\vertex [right=of z1] (y);
				\vertex [above =.1em of z2] (y2) {\(\dot\chi\)};
				\vertex [below =.1em of z1] (y1) {\(\dot\chi\)};
				\diagram { 
					(x) -- [scalar] (z1);
					(z1) -- [scalar,quarter left] (z2);
					(z1) -- [scalar,quarter right] (z2);
					(z1) -- [scalar] (y);
				};
			\end{feynman}			
\end{tikzpicture}
\quad
\begin{tikzpicture}
	\begin{feynman}[inline=(d.base)]
				\vertex (x);
				\vertex [right=of x] (z1);
				\vertex [above=of z1] (z2);
				\vertex [right=of z1] (y);
				\vertex [above =.1em of z2] (y2) {\(\dot\chi\)};
				\vertex [below =.1em of z1] (y1) {\(\dot\chi\)};
				\diagram { 
					(x) -- [scalar] (z1);
					(z1) -- [scalar,quarter left] (z2);
					(z1) -- [scalar,quarter right] (z2);
					(z1) -- [scalar] (y);
				};
			\end{feynman}			
\end{tikzpicture}
\caption{This figure contains the graphs appearing in the graphical expansion of $G(x_1,x_2)$ at second order in $\lambda$ in a $\lambda\Phi^4$ where one internal vertex correspond to $m_\beta^2 \phi^2$ and the other vertex to $\phi^4$.
The dotted lines correspond to the thermal propagator expanded over the Matsubara frequencies \eqref{eq:thermal-propagator} the other lines to the propagator in \eqref{eq:real-time-propagator}. In this case, there is conservation of the Matsubara frequencies at every vertex reached by dotted lines. 
} 
\label{fig:diagrams-phi4}
\end{figure}
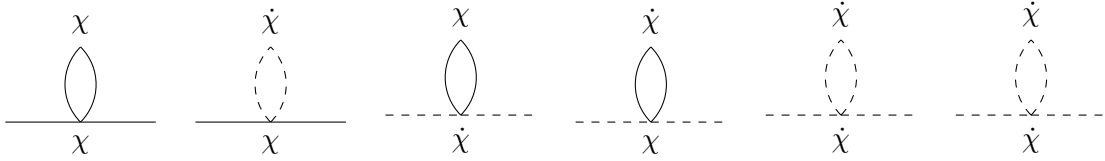

\section{Practical computations in perturbation theory}\label{sec:practical}

In this section we would like to compare some correlation functions in the state $\omega^\beta$ with those obtained in $\omega^{\beta,V}$.
As discussed in the introduction, we shall see that in the large time limit it is necessary to take into account the effect of $U(i\beta)$ present in the state also when the considered observable is very far from the region where the interaction is switched on.
We shall compute the correction to the time-ordered propagator of the theory in two cases, namely when the interaction Lagrangian is quadratic and when it is cubic. 

\subsection{Evaluation of the time-ordered propagator at first order for quadratic interaction}\label{sec:correction-phi2}

In this subsection we consider the case of a quadratic interaction Lagrangian, hence we are interested in evaluating $\omega^{\beta,\lambda V}(R_V(O))$ where
\[
V= \int \chi h \frac{\phi^2}{2}  d\mu\,, \qquad O = T(\Phi(x_1),\Phi(x_2))\,.
\] 
When $x_1=x_2$ this expectation value is also related with the self-energy of a $g\phi^4$ theory -- \textit{cf.} \cite{LeBellac}.
We then compare the expectation values of $R_{\lambda V}(O)$ computed in $\omega^{\beta}$ and in $\omega^{\beta,\lambda V}$,
\[
F:=\omega^{\beta,\lambda V}(R_{\lambda V}(O)),\qquad A :=\omega^{\beta}(R_{\lambda V}(O)), \qquad B:=F-A.
\]
We want to prove that $\lim_{t\to\infty} \lim_{h\to1} B(x_1+te_0,x_2+te_0)$ is not zero at first order in perturbation theory.
Notice that at first order in perturbation theory we have to compute
\begin{equation*}
	B^{[1]}(x_1,x_2) = -
	\int_0^\beta du \;   \omega_\beta^c\left( \Phi(x_1)\Phi(x_2) \otimes \alpha_{iu}\dot{V} \right)=
	\begin{aligned} 
		-\int_0^\beta du 
		\begin{tikzpicture}
			[baseline={([yshift=-.5ex]current bounding box.center)},thick]
			\begin{feynman}[inline=(d.base)]
				\vertex (z);
				\vertex [above left=of z] (x) {\(\Phi(x_1)\)};
				\vertex [above right=of z] (y) {\(\Phi(x_2)\)};
				\vertex [below=0.3em of z] {\(\alpha_{iu} \dot{V}\)};
				
				\diagram* { 
					(x) -- [fermion] (z);
					(y) -- [fermion] (z);
				};
			\end{feynman} 
		\end{tikzpicture}
	\end{aligned}\,,
\end{equation*}
where the arrows denotes the propagator $\Delta^{\beta}_+$. 
In the limit $h\to 1$ and with fixed $\chi$ we obtain an expression invariant under space translation. 
We can thus compute the spatial Fourier transform\footnote{In the following we shall use the following convention regarding the Fourier transform: $B(\mathbf{x})=\frac{1}{(2\pi)^3}\int {d}^3\mathbf{p}\hat{B}(\mathbf{p})^{-i\mathbf{p}\cdot\mathbf{x}}$, $\hat{B}(\mathbf{k})=\int B(\mathbf{x})e^{i\mathbf{p}\cdot\mathbf{x}}$.}, namely the Fourier transform with respect to $\mathbf{x}-\mathbf{y}$. 
We get
\begin{gather*}
\hat{B}^{[1]}\left(t_1,t_2,\mathbf{p}\right) 
= \\
\frac{-2}{(2\pi)^3}
\int_0^\beta du \int dt_3\int dp_0^1dp_0^2 \;  \dot{\chi}(t_3) e^{ ip_0^1 (t_3-t_1)} e^{ ip_0^2 (t_3-t_2)}   e^{-u (p_0^1+p_0^2)} 
\hat\Delta^\beta_+(p_0^1,\mathbf{p})\hat\Delta^\beta_+(p_0^2,\mathbf{p})
\end{gather*}
where $\hat{\Delta}_+(p_0,\mathbf{p})$ is the Fourier transform of the thermal two-point function given in \eqref{eq:thermal-2pt}.
After changing the variable $t=\frac{t_1+t_2}{2}$ and $\delta t = \frac{t_1-t_2}{2}$, in view of the form of
$\hat\Delta^\beta_+$ we can take all the integrals to obtain
\begin{align}\label{eq:Bft}
\hat{B}^{[1]}(t,\delta t,\mathbf{p}) 
&= 
\frac{-2}{(2\pi)^3}\frac{b(w)^2}{4w^2}     
\left(
\left( \hat{\dot{\chi}}(2w) e^{ -i 2w t} + \hat{\dot{\chi}}(-2w)e^{ i 2w t}\right) \frac{1-e^{-2\beta w}}{2w} 
+\beta e^{-\beta w} e^{ -i 2w \delta t}
\right)
\end{align}
where $w=\sqrt{|\mathbf{p}|^2+m^2}$ and where $\hat{\dot\chi}$ denotes the Fourier transform of $\dot{\chi}$.
Consider now 
\[
\mathcal{B}(f; t,\delta t)= \int_{\mathbb{R}^3} d^3\mathbf{p}   \hat{B}^{[1]}(t,\delta t,\mathbf{p}) \hat{f}(\mathbf{p}) 
\]
where $f$ is a compactly supported smooth function.
In the limit $t\to\infty$, due to the Riemann Lebesgue lemma the oscillating factors disappears and we get
\[
\mathcal{B}_\infty(f; \delta t) :=\lim_{t\to\infty} \mathcal{B}(f; t,\delta t) =  
-\frac{2}{(2\pi)^3}\int_{\mathbb{R}^3} d^3\mathbf{p}    \frac{b(w)^2}{4w^2}  \beta e^{-\beta w} e^{ -i 2w \delta t} \hat{f}(\mathbf{p}) 
\]
In particular, the integral kernel of $\mathcal{B}$ seen as operator on $f$ is such that
\[
\tilde{\mathcal{B}}_\infty (\mathbf{x}, \delta t) = 
-\frac{2}{(2\pi)^3}\int_{\mathbb{R}^3} d^3\mathbf{p}    \frac{b(w)^2}{4w^2}  \beta e^{-\beta w} e^{ -i2 w \delta t} e^{ -i \mathbf{p}\cdot\mathbf{x}}\,.
\]
Evaluation at $\mathbf{x}=0$ and $\delta t=0$ leads to
\[
\tilde{\mathcal{B}}_\infty (0,0) = 
-\frac{2}{(2\pi)^3}\int_{\mathbb{R}^3} d^3\mathbf{p}    \frac{b(w)^2}{4w^2}  \beta e^{-\beta w}   
\]
which is a strictly positive quantity because the function which is integrated is smooth and positive on the domain of integration.
Furthermore, it is strictly positive in $\mathbf{p}=0$.

\bigskip
We proceed now analysing
\[
\hat{F}^{[1]}(t,0,\mathbf{p}) = \hat{A}^{[1]}(t,0,\mathbf{p})+\hat{B}^{[1]}(t,0,\mathbf{p}) 
\]
where $A^{[1]} = -i \omega(T(V, O)) +i \omega(V\star O)$ is the contribution of the first two terms and $B^{[1]}=-   \int_0^\beta du  \; \omega^{\beta,c}(O\otimes \dot{V}_{i u})$ was already discussed above.
We have 
\[
\hat{A}^{[1]}(t,0,\mathbf{p}) =
\frac{-i}{(2\pi)^3}
\int dt_y  \theta(t-t_y) \left(\hat{\Delta}_-^2(t_y-t,\mathbf{p})-\hat{\Delta}_+^2(t_y-t,\mathbf{p})\right) \chi(t_y).
\]
Up to renormalization 
\[
\left((\Delta^\beta_-)^2- (\Delta^\beta_+)^2 \right)=
\left(\Delta_-- \Delta_+ \right) \left(\Delta_-+ \Delta_++2W \right)\,,
\]
hence
\[
\hat{A}^{[1]}(t,0,\mathbf{p}) = \frac{-i}{(2\pi)^3} \int_{-\infty}^{t} dt_y \chi(t_y) \frac{1}{4 w^2} \left( e^{i2 w(t_y-t)}-e^{-i2 w(t_y-t)}\right)
(1+2b_-)
\]
where $w=\sqrt{|\mathbf{p}|^2+m^2}$, $b(w) = b_+(w) =(1-e^{-\beta w})^{-1}$ and $b_-(w) = (e^{\beta w}-1)^{-1}$.
Writing $e^{i2w t }$ as $\frac{-i}{2w}\frac{d}{dt}e^{i2wt}$ and integrating by parts, assuming $t>0$ we find
\[
\hat{A}^{[1]}(t,0,\mathbf{p}) = - \frac{1}{(2\pi)^3} \frac{1}{4 w^3} (1+2b_-)
+
\frac{1}{(2\pi)^3} \left(\hat{\dot{\chi}}(2w)e^{-i2 w t}+\hat{\dot{\chi}}(-2w)e^{i2 w t}\right) 
\frac{1}{8 w^3} (1+2b_-)\,.
\]
Notice that the contribution $-\frac{1}{(2\pi)^3} \frac{1}{4 w^3}$ is logarithmically divergent when integrated over $d\mathbf{p}^3$.
This divergence is not present if renormalization is properly taken into account (see \eqref{eq:loop-ren} below) or if $B$ is evaluated at $\mathbf{x}_1-\mathbf{x}_2\neq 0$. 
Recalling the form of $B$ given in \eqref{eq:Bft}
\begin{align*}
\hat{B}^{[1]}(t,0,\mathbf{p}) 
&=-  \frac{1}{(2\pi)^3} \frac{b_+b_-}{2w^2} - \frac{1}{(2\pi)^3} \frac{(b_++b_-)}{8 w^3} 
\left(\hat{\dot{\chi}}(2w)e^{-i2 w t}+\hat{\dot{\chi}}(-2w)e^{i2 w t}\right)\,.
\end{align*}
We observe that both $A$ and $B$ depend on the cut-off function $\chi$, however, as expected, $\hat{F}=A+B$ does not depend on $\chi$
\begin{equation}\label{eq:F2-2pt}
\hat{F}^{[1]}(t,0,\mathbf{p})
 = 
 -  \frac{1}{(2\pi)^3} \left(\frac{b_+b_-}{2w^2}
+ \frac{b_++b_-}{4 w^3} \right).
\end{equation}
We now compare this result with the expectation value of $O$ computed in a state of a free field of mass $\tilde{m}=\sqrt{m^2+\lambda}$.
The latter is $\omega_\la(x_1-x_2):=\omega^{\beta}(O)=\Delta^{\beta,\tilde{m}}_+(x_1-x_2)$ if $x_1 \gtrsim x_2$  where $\Delta^{\beta,\tilde{m}}_+$ is given in \eqref{eq:2pt-thermal-m}. 
Up to first order in $\lambda$, we have that  $w_{\lambda}^{2} = w^2+\lambda = |\mathbf{p}|^2+m^2+\lambda$ and hence
\[
w_\lambda \approx  w + \frac{\lambda}{2 w}, \qquad
\frac{1}{w_\lambda} \approx  \frac{1}{w} - \frac{\lambda}{2 w^3}, \qquad   
b_\pm(w_\lambda) \approx b_\pm(w)-\frac{\lambda}{2 w} b_-b_+ 
\]
hence
\[
\hat\omega_\lambda (t,0,\mathbf{p}) \approx \frac{1}{(2\pi)^3}  \frac{b_+(w_\lambda) +b_-(w_\lambda)}{2w_\lambda} = \frac{1}{(2\pi)^3} \frac{b_+ +b_-}{2w} - \lambda  \frac{1}{(2\pi)^3} \left( \frac{b_+b_-}{2w^2} + \frac{b_++b_-}{4w^3}   \right)
\]
we see that at first order in $\lambda$ $\hat{\omega}_\lambda$ coincides with $\hat{F}^{[1]}$.

\bigskip
We conclude this section observing that in $\hat{F}^{[1]}$ survives a contribution from $B$ -- which is the result one would obtain employing the imaginary time formalism -- and a contribution from $A$, which comes from the real time formalism.
Hence, there are cases where it is necessary to take into account the full form of the state and the real time formalism alone does not furnish complete results. 

\bigskip
Notice that similar corrections are necessary also in a $\lambda \Phi^4$ theory. 
Actually, as discussed in \eqref{eq:phi4-rep} the representation of the interaction potential in $\Fmuc$ is such that
\[
\mathcal{L}_I=\lambda \frac{r_{\Delta_+^\beta}(\Phi^4)}{4!}=\lambda \frac{\phi^4}{4!}  + \lambda m_\beta^2 \frac{\phi^2}{2} 
\]
where
\[
m_\beta^2 = \lim_{x\to0} \left( \Delta_+^{\beta}(x)-\Delta_+^{\infty}(x) \right) = \frac{1}{(2\pi)^3}\int d^3 \mathbf{p} \frac{b_-}{w} 
\]
is the known thermal mass. The presence of a second order contribution in $V$ implies that there are correction to the self-energy similar to the correction of $F$ discussed above.
In particular, the contribution due to the first two graphs depicted in Figure \ref{fig:diagrams-phi4} to the self-energy corresponds to the integral over $\mathbf{p}\in\mathbb{R}^3$ of $\hat{F}^{[1]}$ given in \eqref{eq:F2-2pt} multiplied by $m_\beta^2$. 

\subsection{Evaluation of the time-ordered propagator at second order for cubic interaction}
In this section we show that the second order corrections to the time-ordered product of two interacting fields in a $\lambda\phi^3$ theory induced by $U(i\beta)$ are in general not vanishing if the adiabatic limit and the large time limit are considered. 

To this end we choose
\[
V= \int \chi h \frac{\phi^3}{3!}  d\mu_x, \qquad O = T(\Phi(x_1),\Phi(x_2))\,,
\] 
and we consider
\[
F=\omega^{\beta,\lambda V}(R_{\lambda V}(\alpha_tO)) -\omega^{\beta}(R_{\lambda V}(\alpha_tO))\,. 
\]
We shall see that $\lim_{t\to\infty} \lim_{h\to1} F$ is not zero.

To this end, we first of all notice that $F$ vanishes at order $0$ and $1$ in $\lambda$.
Furthermore, the second order contributions to $F$ are the following terms
\begin{align*}
A &= 
i\int_0^\beta du_1    \omega_\beta^c\left( T(V, \Phi(x_1)\Phi(x_2))  - V\star  \Phi(x_1)\Phi(x_2)   \otimes \alpha_{iu_1}\dot{V} \right)  
\\
B &= 
i\int_0^\beta du_1    \omega_\beta^c\left( \Phi(x_1)\Phi(x_2)   \otimes \alpha_{iu_1} \left(T (V,\dot{V})-V\star \dot{V}\right)\right)  \\
C&= \int_0^\beta du_2\int_0^{u_2} du_1    \omega_\beta^c\left( \Phi(x_1)\Phi(x_2)   \otimes \alpha_{iu_1}\dot{V}\otimes \alpha_{iu_2}\dot{V}\right)
\end{align*}
They have the following graphical expansion in the limit where $h\to1$ and up to a symmetrization in $t_1$,$t_2$

		\begin{equation*}
		\begin{aligned}
		& A =   
		 i \int_0^\beta du \left(
		\begin{tikzpicture}[baseline={([yshift=-.5ex]current bounding box.center)},thick]
			\begin{feynman}[inline=(d.base)]
				\vertex (z1);
				\vertex [right=of z1] (z2);
				\vertex [above left=of z1] (x) {\(\phi\)};
				\vertex [above right=of z2] (y) {\(\phi\)};
				\vertex [left=.1em of z1] {\(V\)};
				\vertex [right=.1em of z2] {\(\alpha_{iu}\dot{V}\)};
				
				\diagram* { 
					(x) -- (z1);
					(z1) -- [fermion, half left](z2);
					(z1) -- [fermion, half right] (z2);
					(y) -- [fermion] (z2);
				};
			\end{feynman}
		\end{tikzpicture} -
		\begin{tikzpicture}[baseline={([yshift=-.5ex]current bounding box.center)},thick]
			\begin{feynman}[inline=(d.base)]
				\vertex (z1);
				\vertex [right=of z1] (z2);
				\vertex [above left=of z1] (x) {\(\phi\)};
				\vertex [above right=of z2] (y) {\(\phi\)};
				\vertex [left=.1em of z1] {\(V\)};
				\vertex [right=.1em of z2] {\(\alpha_{iu}\dot{V}\)};
				
				\diagram* { 
					(x) -- [anti fermion] (z1);
					(z1) -- [fermion, half left](z2);
					(z1) -- [fermion, half right] (z2);
					(y) -- [fermion] (z2);
				};
			\end{feynman} 
		\end{tikzpicture} \right)  \\
		B &=  i  \int_0^\beta du  \left(
		\begin{tikzpicture}[baseline={([yshift=-.5ex]current bounding box.center)},thick]
			\begin{feynman}[inline=(d.base)]
				\vertex (z1);
				\vertex [right=of z1] (z2);
				\vertex [above left=of z1] (x) {\(\phi\)};
				\vertex [above right=of z2] (y) {\(\phi\)};
				\vertex [left=.1em of z1] {\(\alpha_{iu} V\)};
				\vertex [right=.1em of z2] {\(\alpha_{iu}\dot{V}\)};
				
				\diagram* { 
					(x) -- [fermion] (z1);
					(z1) -- [half left](z2);
					(z1) -- [half right] (z2);
					(y) -- [fermion] (z2);
				};
			\end{feynman} 
		\end{tikzpicture} -
		\begin{tikzpicture}[baseline={([yshift=-.5ex]current bounding box.center)},thick]
			\begin{feynman}[inline=(d.base)]
				\vertex (z1);
				\vertex [right=of z1] (z2);
				\vertex [above left=of z1] (x) {\(\phi\)};
				\vertex [above right=of z2] (y) {\(\phi\)};
				\vertex [left=.1em of z1] {\(\alpha_{iu} V\)};
				\vertex [right=.1em of z2] {\(\alpha_{iu}\dot{V}\)};
				
				\diagram* { 
					(x) -- [fermion] (z1);
					(z1) -- [fermion, half left](z2);
					(z1) -- [fermion, half right] (z2);
					(y) -- [fermion] (z2);
				};
			\end{feynman} 
		\end{tikzpicture} \right) \\
		C &= 
		\int_0^\beta du_2 \int_0^{u_2} du_1 
		\begin{tikzpicture}[baseline={([yshift=-.5ex]current bounding box.center)},thick]
			\begin{feynman}[inline=(d.base)]
				\vertex (z1);
				\vertex [right=of z1] (z2);
				\vertex [above left=of z1] (x) {\(\phi\)};
				\vertex [above right=of z2] (y) {\(\phi\)};
				\vertex [left=.1em of z1] {\(\alpha_{iu_1} \dot{V}\)};
				\vertex [right=.1em of z2] {\(\alpha_{iu_2}\dot{V}\)};
				
				\diagram* { 
					(x) -- [fermion] (z1);
					(z1) -- [fermion, half left](z2);
					(z1) -- [fermion, half right] (z2);
					(y) -- [fermion] (z2);
				};
			\end{feynman} 
		\end{tikzpicture}.
		\end{aligned}
	\end{equation*}
	where the plane lines correspond to $\Delta^\beta_F$ and the lines with arrows to $\Delta^\beta_+$.  
We denote by $t_3$ and $t_4$ the time of the two internal vertices, by $p^2$ the internal four momentum of the loop and by $p^1$ and $p^3$ the external four momenta. With this conventions, and in view of the spatial momentum conservation at the internal vertices when the spatial cut-off is removed, we have that the spatial Fourier transform of $A,B,C$ are such that
\begin{align}
	\nonumber
	\hat{A}(t_1,t_2,\mathbf{p}) =
	&\frac{1}{(2\pi)^6} i\mathcal{S}\int_0^\beta du \int dt_3\int dt_4\int dp_0^1dp_0^2dp_0^3
	\chi(t_3)\dot{\chi}(t_4) e^{ ip_0^1 (t_3-t_1)} e^{ ip_0^3 (t_4-t_2)}  e^{ ip_0^2 (t_4-t_3)} \\ \nonumber 
	&\cdot e^{-u ( p_0^3+p_0^2)} 
	\left(\hat{\Delta}_F^\beta(-p_0^1,\mathbf{p})-\hat{\Delta}_+^\beta(-p_0^1,\mathbf{p})\right)(\hat{\Delta}^\beta_+)^2(p_0^2,\mathbf{p})\hat{\Delta}^\beta_+(p_0^3,\mathbf{p})
	\\ 
	\nonumber
	\hat{B}(t_1,t_2,\mathbf{p}) =&\frac{1}{(2\pi)^6} i\mathcal{S}\int_0^\beta du \int dt_3\int dt_4\int dp_0^1dp_0^2dp_0^3   \dot{\chi}(t_3)\chi(t_4) 
	e^{ ip_0^1 (t_3-t_1)} e^{ ip_0^3 (t_4-t_2)}  e^{ ip_0^2 (t_4-t_3)} 
	\\ \nonumber
	&\cdot e^{-u ( p_0^3+p_0^2)}e^{-u ( p_0^1-p_0^2)} 
	\Delta_+^\beta(p_0^1,\mathbf{p})\left((\hat{\Delta}^\beta_F)^2(p_0^2,\mathbf{p})-(\hat{\Delta}^\beta_+)^2(p_0^2,\mathbf{p})\right)\hat{\Delta}^\beta_+(p_0^3,\mathbf{p})
	\\
	\label{eq:B0}
	\hat{C}(t_1,t_2,\mathbf{p}) =&\frac{1}{(2\pi)^6} \mathcal{S}\int_0^\beta du_2\int_0^{u_2} du_1 \int dt_3\int dt_4\int dp_0^1dp_0^2dp_0^3   \dot\chi(t_3)\dot{\chi}(t_4) \\ \nonumber
	&\cdot e^{ ip_0^1 (t_3-t_1)} e^{ ip_0^3 (t_4-t_2)}  e^{ ip_0^2 (t_4-t_3)} e^{-u_1 ( p_0^1-p_0^2)}e^{-u_2 ( p_0^3+p_0^2)} 
	\hat{\Delta}_+^\beta(p_0^1,\mathbf{p})(\hat{\Delta}^\beta_+)^2(p_0^2,\mathbf{p})\hat{\Delta}^\beta_+(p_0^3,\mathbf{p})
\end{align}
where $\mathcal{S}$ is the operator which realizes the symmetrization in $t_1,t_2$ namely $\mathcal{S}f(t_1,t_2)=\frac{1}{2}f(t_1,t_2)+\frac{1}{2}f(t_2,t_1)$.
Let us discuss $C$.
The integrals over $u_i$ and $t_3,t_4$ can be directly taken.
We then obtain
\begin{align*}
\hat{C}(t_1,t_2,\mathbf{p}) &=\frac{1}{(2\pi)^6} \mathcal{S}\int dp_0^1dp_0^2dp_0^3  
e^{ -ip_0^1 t_1} e^{ -ip_0^3 t_2}  
 \hat{\dot\chi}(p_0^1-p_0^2)\hat{\dot{\chi}}(p_0^3+p_0^2) f(p_0^1,p_0^2,p_0^3,\mathbf{p})
\end{align*}
where
\[
f(p_0^1,p_0^2,p_0^3,\mathbf{p})
=
 \left(\frac{1-e^{-\beta ( p_0^3+p_0^2) }}{(p_0^1-p_0^2)(p_0^3+p_0^2)}
-\frac{1-e^{-\beta ( p_0^1+p_0^3)}}{(p_0^1-p_0^2)(p_0^1+p_0^3)}\right)
\hat{\Delta}_+^\beta(p_0^1,\mathbf{p})(\hat{\Delta}^\beta_+)^2(p_0^2,\mathbf{p})\hat{\Delta}^\beta_+(p_0^3,\mathbf{p})\,.
\]

We now consider the limit $t_1+t_2\to\infty$ of $C$. 
To this end, we observe that in view of the form on $\hat{\Delta}^\beta_+$, $p^1=(p_0^1,\mathbf{p})$ and $p^3=(p_0^3,\mathbf{p})$ are forced to be supported on the positive or negative mass shells. 
Since $\hat{\dot\chi}$ is of rapid decrease, after computing the integral over $p_0^2$ and exploiting the form of $(\hat{\Delta}_+^\beta)^2$, we get that when both $(p_0^1,\mathbf{p})$ and $(p_0^3,\mathbf{p})$ are future directed or both are past directed they cannot contribute to $C$, because the corresponding contribution vanishes by Riemann Lebesgue lemma.
Hence, only combinations of future/past directed momenta can contribute to $C$ in the large time limit.
Therefore,
\[
\lim_{t\to\infty} \int d^3\mathbf{p}\;  \hat{C}(t+ \delta t,t-\delta t,\mathbf{p}) = 
\int d^3\mathbf{p} \; \hat{C}_\infty( \delta t,\mathbf{p})\,,
\]
where
\begin{align}
	\hat{C}_\infty( \delta t,\mathbf{p}) &=
	-\frac{1}{(2\pi)^6}\mathcal{S}\int dp_0^2
	 e^{ -i2w\delta t} |\hat{\dot\chi}(w-p_0^2)|^2
	 \left(
	 	\frac{1-e^{-\beta (p_0^2-w) }}{(p_0^2-w)^2}-\frac{\beta}{(p_0^2-w)}
	\right)
	\frac{b(w)^2}{4w^2} e^{-\beta w}
	(\hat{\Delta}^\beta_+)^2(p_0^2,\mathbf{p})
	\nonumber
	\\
	&-\frac{1}{(2\pi)^6}\mathcal{S}\int dp_0^2
	e^{ i2w\delta t} |\hat{\dot\chi}(w+p_0^2)|^2
	\left(
		\frac{1-e^{-\beta (p_0^2+w) }}{(p_0^2+w)^2}-\frac{\beta}{(p_0^2+w)}
	\right)
	\frac{b(w)^2}{4w^2} e^{-\beta w}
	(\hat{\Delta}^\beta_+)^2(p_0^2,\mathbf{p})\,.
	\label{eq:C1}
\end{align}
We now analyse $A$. To this end we notice that the integral over $u$ and $t_4$ can be directly performed.
Before computing the integral over $t_3$, we recall that $\Delta_F^\beta-\Delta_+^\beta = i\Delta_A $ and hence
\[
\int dp e^{-ipt} (\hat{\Delta}_F^\beta(-p,\mathbf{p})-\hat{\Delta}_+^\beta(-p,\mathbf{p}))
= i \theta(t) \int dp e^{-ipt}\hat{\Delta}(p,\mathbf{p})
\]
and hence 
\begin{align*}
\int dt_3 e^{ -ip_0^2 t_3} \chi(t_3) \int dp_0^1 e^{ ip_0^1 (t_3-t_1)} \left(\hat{\Delta}_F^\beta(-p_0^1,\mathbf{p})-\hat{\Delta}_+^\beta(-p_0^1,\mathbf{p})\right) \\
=
i\int dt_3 e^{ -ip_0^2 t_3}\chi(t_3) \theta(t_1-t_3) \int dp_0^1 e^{ ip_0^1 (t_3-t_1)} \hat{\Delta}(p_0^1,\mathbf{p})\\
=
i
\left(\int dp_0^1 \frac{ \chi(t_1)}{i(p_0^1-p_0^2)}e^{ -ip_0^2 t_1} - 
 \frac{\hat{\dot{\chi}}(p_0^1-p_0^2)}{i(p_0^1-p_0^2)}  e^{ -ip_0^1t_1} 
 \right)\hat{\Delta}(p_0^1,\mathbf{p})\,,
\end{align*}
where we have integrated by parts. 
In the limit of large $t_1$,  $\chi(t_1)=1$, furthermore, the contribution proportional to $\chi(t_1)=1$ vanishes because of the Riemann Lebesgue lemma. Actually, $\hat{\Delta}_+^\beta(p_0^3,\mathbf{p})(\hat{\Delta}_+^\beta)^2(p_0^2,\mathbf{p})$ vanishes on the points where $p_0^2+p_0^3=0$. 

Hence, in the limit $t\to+\infty$ we have 
\begin{align}
\hat{A}_\infty(\delta t,\mathbf{p}) 
&=\frac{1}{(2\pi)^6} \mathcal{S}\int dp_0^2  e^{ -i2w \delta t}  |\hat{\dot{\chi}}(p_0^2-w)|^2  \frac{1-e^{-\beta ( p_0^2-w)}}{(p_0^2-w)^2}
\frac{b(w)e^{-\beta w}}{4w^2}
(\hat{\Delta}^\beta_+)^2(p_0^2,\mathbf{p}) 
\notag
\\
& -\frac{1}{(2\pi)^6}\mathcal{S}\int dp_0^2  e^{ i2w \delta t}  |\hat{\dot{\chi}}(p_0^2+w)|^2  \frac{1-e^{-\beta ( p_0^2+w)}}{(p_0^2+w)^2}
\frac{b(w)}{4w^2}
(\hat{\Delta}^\beta_+)^2(p_0^2,\mathbf{p}) 
\label{eq:A1}
\end{align}

\bigskip
To compute $B$ we need to preliminarily analyse 
\begin{align}
{\Delta^\beta_F(x)}^2 -{\Delta^\beta_+(x)}^2 &= {(\Delta_F(x)+W_\beta(x))}^2 -{(\Delta_+(x)+W_\beta(x))}^2
\notag
\\
 &= \Delta_F(x)^2-\Delta_+(x)^2 + 2i \Delta_A(x)W_\beta(x) 
 \notag
 \\
 &= (\Box+a)\int_{4m^2}^\infty dM^2 \frac{\rho_2(M)}{M^2+a} i\Delta_A(x;M)  
 -2i \theta(-t_x) \Delta(x)W_\beta(x) 
\label{eq:loop-ren}
\end{align}
where we have used the K\"allen-Lehmann representation to represent and regularize $\Delta_F(x)^2$ and $\Delta_+(x)^2$, hence $a$ is a parameter which takes into account the renormalization freedom in the definition of $\Delta_F^2$.
Moreover,
\[
\rho_2(M) = \frac{1}{16\pi^2} \sqrt{1-\frac{4m^2}{M^2}}
\]
and $\Delta_A(x;M)$ is the advanced fundamental solution of the Klein Gordon equation with mass $M$.
Outside $x=0$, 
\[
\Xi=(\Box+a)\int_{4m^2}^\infty dM^2 \frac{\rho_2(M)}{M^2+a} i\Delta_A(x;M) = ic\delta(x)-\theta(-t_x) 
\int_{4m^2}^\infty dM^2 {\rho_2(M)} i\Delta(x;M).
\]
where  $c$ is a renormalization constant.
Choosing $c=0$ we have
\begin{align}
\Xi &= -\theta(-t_x)
\int dp_0d\mathbf{p} e^{ip_0t_x}e^{-i\mathbf{p}\cdot \mathbf{x}}\int_{4m^2}^\infty dM^2 {\rho_2(M)} \delta(p_0^2-|\mathbf{p}|^2-M^2) \text{sign}(p_0)
\notag
\\
&= -\theta(-t_x) 
\int dp_0d\mathbf{p}e^{ip_0t_x}e^{-i\mathbf{p}\cdot \mathbf{x}}
\theta(p_0^2-|\mathbf{p}|^2-4m^2) 
 {\rho_2\left(\sqrt{p_0^2-|\mathbf{p}|^2}\right)} \text{sign}(p_0)
\label{eq:T}
\end{align}
The sum of $\Xi$ and $-2i \theta(-t_x) \Delta(x)W_\beta(x) $ can be written as  $-i\theta(-t) Q(t)$ where $Q$ is antisymmetric and it is
\begin{equation}\label{eq:Q}
Q(x) = 
\left(  \int_{4m^2}^\infty dM^2 {\rho_2(M)} \Delta(x;M)        +   2\Delta(x)W_\beta(x)   \right).
\end{equation}
Its Fourier transform can be easily computed, hence when we integrate in $t_3$ and $t_4$ we have to take into account the presence of the Heaviside step function.
However, after simmetryzing in the external points and using the antisymmetry of $Q(t,\mathbf{p})$ for $t\to-t$ we can tame the presence of the Heaviside step functions. Actually on the points where $p_0^3=-p_0^1$ (those which survives under the limit $t\to\infty$) we have that
\begin{gather*}
\int dt_3\int dt_4  \dot{\chi}(t_3)\chi(t_4) 
e^{ i(p_0^3+p_0^2) (t_4-t_3)}  \theta(t_3-t_4) 
-
\int dt_3\int dt_4  \dot{\chi}(t_4)\chi(t_3) 
 e^{ i(p_0^3+p_0^2) (t_4-t_3)}  \theta(t_4-t_3) 
\\ 
= \frac{1}{i(p_0^3+p_0^2)}
-
\frac{|\hat{\dot{\chi}}(p_0^2+p_0^3)|^2}{i(p_0^3+p_0^2)}.
\end{gather*}
With these observations, the limit $t=t_1+t_2\to\infty$ can now be taken in $\hat{B}(t_1,t_2,\mathbf{p})$ given in \eqref{eq:B0}, leading to
\begin{align}
\hat{B}_\infty(\delta t,\mathbf{p})
&= \frac{1}{(2\pi)^3}
\beta \mathcal{S}\int dp_0^2 e^{ -i2w \delta t}  
\left(
\frac{1}{(p_0^2-w)}
-
\frac{|\hat{\dot{\chi}}(p_0^2-w)|^2}{(p_0^2-w)} 
\right)
\frac{b(w)^2e^{-\beta w}}{4w^2}
\hat{Q}(p_0^2,\mathbf{p})
\notag
\\
&+\frac{1}{(2\pi)^3} \beta \mathcal{S}\int dp_0^2 e^{ i2w \delta t}  
\left(
\frac{1}{(p_0^2+w)}
-
\frac{|\hat{\dot{\chi}}(p_0^2+w)|^2}{p_0^2+w} 
\right)
\frac{b(w)^2e^{-\beta w}}{4w^2}
\hat{Q}(p_0^2,\mathbf{p})
\label{eq:B1}
\end{align}
The contribution of the renormalization freedom is given by
\begin{align*}
	\hat{B}^c(t_1,t_2,\mathbf{p}) =\frac{1}{(2\pi)^3}
	c2\mathcal{S}\int_0^\beta du \int dt_3\int dp_0^1dp_0^3   \dot{\chi}(t_3)\chi(t_3) e^{ ip_0^1 (t_3-t_1)} e^{ ip_0^3 (t_3-t_2)}   e^{-u ( p_0^3+p_0^1)} 
	\hat{\Delta}_+^\beta(p_0^1,\mathbf{p})\hat{\Delta}^\beta_+(p_0^3,\mathbf{p})\,,
\end{align*}
so that in the limit $t\to\infty$ we have
\begin{align}\label{eq:B1c}
\hat{B}^c_{\infty}(\delta t,\mathbf{p}) = \frac{1}{(2\pi)^3}c \beta  \frac{b(w)^2}{4w^2}e^{-\beta w}.
\end{align}
We observe that $c$ cannot depend on $\beta$.

\bigskip
We thus have that the contribution to $F^{[2]}$ which survives the limit $t\to\infty$ is
\[
\hat{F}^{[2]}_\infty(\delta t,\mathbf{p}) = \hat{A}_\infty(\delta t,\mathbf{p}) +\hat{B}_\infty(\delta t,\mathbf{p}) +\hat{C}_\infty(\delta t,\mathbf{p}) +\hat{B}^c_\infty(\delta t,\mathbf{p}) 
\]
where $\hat{A}_\infty$, $\hat{B}_\infty$, $\hat{C}_\infty$, $\hat{B}^c_\infty$, are respectively as in \eqref{eq:A1}, \eqref{eq:B1}, \eqref{eq:C1}, \eqref{eq:B1c}.  In order to proceed with our analysis we have to evaluate $(\hat{\Delta}_+^\beta)^2(p_0^2,\mathbf{p})$ in $\hat{A}_\infty$ and $\hat{C}_\infty$ and 
$\hat{Q}(p_0^2,\mathbf{p})$ in $\hat{B}_\infty$. 

\bigskip
At this point, we would like to estimate $\hat{F}^{[2]}_\infty$ for $\delta t=0$ and for $\mathbf{p}=0$. We start from the analysis of the $B$ contribution.
We thus need an expression for $\hat{Q}(p_0,\mathbf{p})$ to be inserted in \eqref{eq:B1}.
We decompose $Q$ given in \eqref{eq:Q} as $Q=Y+U$,  
where 
\[
U(x) = 2\Delta(x)W_\beta(x)  
\]
and its Fourier transform is such that
\begin{gather*}
\hat{U}(p_0,\mathbf{p}) = \frac{1}{(2\pi)^6}2\int_{\mathbb{R}^3} d\mathbf{q}_1\int_{\mathbb{R}^3}d\mathbf{q}_2 \frac{b(w_2)}{4w_1w_2} \delta(\mathbf{p}-\mathbf{q}_1-\mathbf{q}_2) \\\
e^{-\beta w_2}\left( \delta(p_0-w_1-w_2)-\delta(p_0+w_1-w_2)
+\delta(p_0-w_1+w_2)-\delta(p_0+w_1+w_2) \right)
\end{gather*}
we observe that this is an odd function of $p_0^2$.
At the same time, recalling that the expression of $Y$ is similar to  that of $\Xi$  in \eqref{eq:T} where the Heaviside step function is removed, we have that the Fourier transform of $Y$ is 
\begin{align*}
\hat{Y}(p_0,\mathbf{p}) = -
\theta(p_0^2-|\mathbf{p}|^2-4m^2) 
 {\rho_2(p_0^2-|\mathbf{p}|^2)} \text{sign}(p_0)
\end{align*}
which is an odd function of $p_0$. Hence 
\[
\hat{Q}(p_0,\mathbf{p}) = \hat{Y}(p_0,\mathbf{p})+\hat{U}(p_0,\mathbf{p})
\]
is an odd function of ${p}_0$. 
This implies that $\hat{B}_\infty(\delta t,\mathbf{p}) = 0$ 
because, after symmetrization, the integrand in $p_0^2$ present in 
\eqref{eq:B1}
 is a multiplication of an even function with and an odd function of $p_0^2$, therefore the integral vanishes.
The contribution of the graph $B$ in the limit $t\to\infty$ is thus made by the renormalization freedom \eqref{eq:B1c} only. 

\bigskip
To analyse the contribution of $A$ and $C$ at $\delta t=0$ and $\mathbf{p}=0$ we need $(\hat{\Delta}^\beta_+)^2(p_0,\mathbf{p})$. 
Notice that (up to a constant)
\begin{gather*}
(\hat{\Delta}^\beta_+)^2(p_0,\mathbf{p}) = \frac{1}{(2\pi)^6}\int_{\mathbb{R}^3} d\mathbf{q}_1\int_{\mathbb{R}^3}d\mathbf{q}_2 \frac{b(w_1)b(w_2)}{4w_1w_2} \delta(\mathbf{p}-\mathbf{q}_1-\mathbf{q}_2) \\\
\left( \delta(p_0-w_1-w_2)+\delta(p_0+w_1-w_2)e^{-\beta w_1}+\delta(p_0-w_1+w_2)e^{-\beta w_2}+\delta(p_0+w_1+w_2)e^{-\beta (w_1+w_2)} \right)
\end{gather*}
where $w_i=\sqrt{\mathbf{q}_i^2+m^2}$. 
Hence, after substituting the previous expression in \eqref{eq:A1} and in \eqref{eq:C1}, setting $\mathbf{p}=0$ we have that the integral over $\mathbf{q}_2$ can be taken and it forces $w_2=w_1=m$. Using ordinary spherical coordinates centred in the origin for $\mathbf{q}_1$ we see that 
the integral over the angular coordinates can be taken because at $\mathbf{p}=0$ the function which is integrated is spherically symmetric. 
The integral over the radial coordinates can be reparametrized with $M=2\sqrt{|\mathbf{q}_1|^2+m^2}$.
Proceeding in this way we obtain $\hat{A}_\infty(0,0)+\hat{C}_\infty(0,0)$ and recalling \eqref{eq:B1c} we obtain
\begin{align*}
\hat{F}^{[2]}_\infty(0,0) 
=& 
\frac{1}{(2\pi)^6}c  \beta  \frac{b(m)^2}{4m^2}e^{-\beta m} 
+
\frac{1}{(2\pi)^6}\frac{\pi}{8} 
\frac{e^{-\beta m}}{m^2} 
{b(m)}
\int_{2m}^\infty dM  \sqrt{1-\frac{4m^2}{M^2}}b\left(\frac{M}{2}\right)^2
\notag
\\
&\cdot\left[
\beta b(m)(1-e^{-\beta M})
\left(\frac{|\hat{\dot\chi}(M-m)|^2}{M-m}+\frac{|\hat{\dot\chi}(M+m)|^2}{M+m}\right)
\right.
\notag
\\
&\left.
+(e^{-\beta m}-e^{-\beta M}) 
\left(\frac{|\hat{\dot\chi}(M-m)|^2}{(M-m)^2}
\right)
+(e^{-\beta (M+m)}-1) 
\left(\frac{|\hat{\dot\chi}(M+m)|^2}{(M+m)^2}
\right)
\notag
\right.
\\
&\left.
-2e^{-\beta \frac{M}{2}}\frac{|\hat{\dot\chi}(m)|^2}{m^2}
  \left(
  \frac{1}{b(m)}
\right)
\right]
\end{align*}
From this analysis and from the expression of $\hat{F}^{[2]}_\infty(\delta t,0)$ we can draw these conclusions:
\begin{enumerate}
\item In the limit $\beta\to\infty$ $F_\infty^{[2]}$ vanishes. 
\item $F$ depends on $\hat{\dot{\chi}}$, more precisely in $\hat{F}_\infty^{[2]}$ there are two contributions, one which depends on $\hat{\dot{\chi}}$ and another one which depends on the choice of the renormalization constant $c$. 
In particular, the renormalization constant can be chosen in such a way that $c = 0$.

\item Notice that in the limit $\chi(t)\to\theta(t)$, $\hat{\dot{\chi}}$ tends to a constant and the integral in $dM$ diverges.

\item We analyse now the limit where $\chi\to1$. 
To this end, consider a smooth $\chi$ with $\chi(t)=1$ for $t\geq 0$ and $\chi(t)=0$ for $t<1$. Consider now the family of time cut off functions $\chi_n(t)=\chi\left(\frac{t}{n}\right)$. In the limit for large $n$, ${\chi}_n$ tends to the constant, while its derivative  $\dot{\chi}_n$ tends to $0$. 
At the same time its Fourier transform is $\hat{\dot{\chi}}_n(p) = \hat{\dot{\chi}}( n p)$. Hence in the limit $n\to\infty$, $\hat{\dot{\chi}}_n(p)$ vanishes for $p\neq0$ because $\hat{\dot{\chi}}$ is a rapidly decreasing function. 
This implies that in the limit $n\to\infty$ $\hat{F}^{[2]}_\infty(0,0)$ vanishes if $c=0$.

\end{enumerate}

In particular, while $\omega^{\beta, V}$ does not depend on $\chi$ because of \cite{FredenhagenLindner}, the analysis presented here shows that the real time formalism -- which corresponds to considering $\omega^\beta$ and ignoring $U(i\beta)$ -- is very sensitive to the way in which the cut-off function is chosen. 
We also notice that even if the limit $n\to\infty$ is taken the analysis done in the real time formalism differs by the one done with  $\omega^{\beta, V}$ by the choice of the renormalization freedom. 
However the principle of general covariance \cite{BFV03, HW01} requires that the renormalization constant $c$ is not state dependent hence it cannot depend on $\beta$, and actually $c$ can be  fixed in the vacuum theory.
Therefore, whenever $c\neq0$, we see a difference between the computation done in the real time formalism with the one in the full equilibrium state   $\omega^{\beta, V}$.

\subsection*{Acknowledgments} 
The work of J. Braga de G\'oes Vasconcellos is supported in part by the National Group of Mathematical Physics (GNFM - INdAM).


\begin{thebibliography}{999}

\bibitem[Al90]{Altherr}
T.~Altherr {\it ``Infrared problem in $g\phi^4$ theory at finite temperature,''}
Phys.\ Lett.\ {\bf B 238}, 360 (1990)

\bibitem[Ar73]{Araki}
H.~Araki,
{\it ``Relative Hamiltonian for faithful normal states of a von Neumann algebra,''} 
Publ.\ RIMS, Kyoto Univ.\ {\bf 9} (1), 165-209 (1973)




\bibitem[BKR78]{BKR}
O.~Bratteli, A.~Kishimoto,  D.W.~Robinson, 
{\it ``Stability properties and the KMS condition,''} 
Commun.\ Math.\ Phys.\ {\bf 61}, 209-238 (1978)

\bibitem[BR87]{BR1}
O.~Bratteli,  D.W.~Robinson, 
{\it ``Operator algebras and quantum statistical mechanics 1,''} 
Springer, Berlin (1987)

\bibitem[BR97]{BR}
O.~Bratteli,  D.W.~Robinson, 
{\it ``Operator algebras and quantum statistical mechanics 2,''} 
Springer, Berlin (1997)



\bibitem[BDF09]{BDF09}
  R.~Brunetti, M.~Duetsch and K.~Fredenhagen,
 {\it ``Perturbative Algebraic Quantum Field Theory and the Renormalization Groups,''}
  Adv.\ Theor.\ Math.\ Phys.\  {\bf 13}, 1541 (2009)
  
  

\bibitem[BF00]{BF00} 
  R.~Brunetti and K.~Fredenhagen,
{\it ``Microlocal analysis and interacting quantum field theories: Renormalization on physical backgrounds,''}
  Commun.\ Math.\ Phys.\  {\bf 208}, 623 (2000)

\bibitem[BFK95]{BFK}
  R.~Brunetti, K.~Fredenhagen and M.~K\"ohler,
{\it ``The microlocal spectrum condition and Wick polynomials of free fields on
  curved spacetimes,''}
  Commun.\ Math.\ Phys.\  {\bf 180}, 633 (1996)
  
\bibitem[BFV03]{BFV03} 
  R.~Brunetti, K.~Fredenhagen and R.~Verch,
{\it ``The Generally covariant locality principle: A New paradigm for local quantum field theory,''}
  Commun.\ Math.\ Phys.\  {\bf 237}, 31 (2003)


\bibitem[CH88]{CH88}
E.~Calzetta and B.~L.~Hu, 
{\it ``Nonequilibrium quantum fields: Closed-time-path effective action, Wigner function, and Boltzmann equation, ''} Phys. Rev. {\bf D 37}, 2878-2900 (1988)
				
				
\bibitem[CM05]{CM05}
M.~E.~Carrington and S.~Mrowczynski 
{\it ``Transport theory beyond binary collisions,''}.
Phys. Rev. {\bf D 71}, 065007 (2005)

\bibitem[CF09]{CF}
B.~Chilian, K.~Fredenhagen,
\emph{``The time-slice axiom in perturbative quantum field theory on globally hyperbolic spacetimes,''}
Commun.\ Math.\ Phys.\ {\bf 287},  513-522 (2009) 


\bibitem[DD16]{DappiaggiDrago}
C.~Dappiaggi and N.~Drago,
{\it ``Constructing Hadamard states via an extended M\o ller operator,''}
Lett.\ Math. Phys. {\bf 106} (11), 1587-1615 (2016)


\bibitem[Dr19]{Drago}
N.~Drago, 
{\it ``Thermal State with Quadratic Interaction,''}
Ann. Henri Poincar\'e {\bf 20}, 905-927 (2019)


\bibitem[DFP18]{DFP-18}
N.~Drago, F.~Faldino, N.~Pinamonti, 
{\it ``On the stability of KMS states in perturbative algebraic quantum field theories,''}
Commun.\ Math.\ Phys.\ {\bf 357} 267-293 (2018)

\bibitem[DFP19]{DFP-19}
N.~Drago, F.~Faldino, N.~Pinamonti, 
``\textit{Relative Entropy and Entropy Production for Equilibrium States in pAQFT}",
Ann. Henri Poincar\'e (2018) 19: 3289. 

\bibitem[DHP16]{DHP}
N.~Drago, T.-P.~Hack, N.~Pinamonti, 
{\it ``The generalised principle of perturbative agreement and the thermal mass,''}
Ann.\ Henri.\  Poincar\'e {\bf 18}, 807-868 (2017)

\bibitem[DFKR14]{DFKR}
M.~Duetsch, K.~Fredenhagen, K.~J.~Keller and K.~Rejzner,
{\it ``Dimensional regularization in position space and a Forest Formula for Epstein-Glaser renormalization,''}
J.\ Math.\ Phys.\ {\bf 55}, 122303 (2014)


\bibitem[EG73]{EG}
H.~Epstein and V.~Glaser, 
{\it ``The role of locality in perturbation theory,''}
Ann.\ Inst.\ Henri\ Poincar\'e Section
A, vol. XIX, n.3, {\bf 211} (1973)


\bibitem[FL14]{FredenhagenLindner} 
  K.~Fredenhagen and F.~Lindner,
{\it ``Construction of KMS States in Perturbative QFT and Renormalized Hamiltonian Dynamics,''}
  Commun.\ Math.\ Phys.\  {\bf 332}, 895 (2014)



\bibitem[FK15]{FredenhagenRejzner} 
  K.~Fredenhagen and K.~Rejzner,
{\it ``Perturbative algebraic quantum field theory,''}
In: Calaque D., Strobl T. (eds) Mathematical Aspects of Quantum Field Theories. Mathematical Physics Studies. Springer, Cham
(2015)



\bibitem[FR87]{Fulling}
S.~A.~Fulling and S.~N.~M.~Ruusenaars, 
{\it ``Temperature, periodicity and horizons,''}
Phys. Rep. {\bf 152}, 135-176 (1987)






\bibitem[HW01]{HW01} 
  S.~Hollands and R.~M.~Wald,
 {\it ``Local Wick polynomials and time-ordered products of quantum fields in curved space-time,''}
  Commun.\ Math.\ Phys.\  {\bf 223}, 289 (2001)


\bibitem[HW02]{HW02} 
  S.~Hollands and R.~M.~Wald,
 {\it ``Existence of local covariant time-ordered products of quantum fields in curved space-time,''}
  Commun.\ Math.\ Phys.\  {\bf 231}, 309 (2002)

\bibitem[HW05]{HW05} 
S.~Hollands and R.~M.~Wald,
{\it ``Conservation of the stress tensor in interacting quantum field theory in curved spacetimes,''}
Rev.\ Math.\ Phys.\  {\bf 17}, 227 (2005)


\bibitem[Ho03]{Hormander}
L.~H\"ormander
{\it ``The Analysis of Linear Partial Differential Operators I,''}
Springer Berlin (2003)

\bibitem[Ka84]{Ka84}
K.~Kajantie, (ed).
{\it ``Quark Matter '84 ''} 
Springer, Berlin, (1985)

\bibitem[KG06]{Kapusta}
J.~I.~Kapusta and C.~Gale 
{\it ``Finite-Temperature Field Theory - Principles and Applications''}
Cambdridge University Press (2006)


	


\bibitem[Kel65]{Kel65}
L.~V.~Keldysh,  
{\it ``Diagram Technique for Nonequilibrium Processes,''}
JETP {\bf 20}, 1018 (1965).

\bibitem[KMS93]{Kolar-Michor-Slocvak-93}
I. Kol\'a\v r, P. W. Michor, J. Slov\'ak,
{\it ``Natural Operations in Differential Geometry,''} (1993)
Springer, 434p.



\bibitem[LW97]{Landsman}
N.~P.~Landsman, C.~G.~van Weert,
{\it ``Real and imaginary time field theory at finite temperature and density,''} 
Phys.\ Rep.\  {\bf 145}, 141 (1987) 

\bibitem[Le00]{LeBellac} 
M.~Le Bellac, 
{\it ``Thermal Field Theory,''}
Cambridge University Press, Cambridge (2000)

\bibitem[Li13]{Lindner}
F.~Lindner,
{\it ``Perturbative Algebraic Quantum Field Theory at Finite Temperature,''}
PhD thesis, University of Hamburg (2013). 

\bibitem[Mat55]{Mat55}
T.~Matsubara,
{\it ``A New Approach to Quantum-Statistical Mechanics,''}
Prog. Theo. Phys. {\bf 14}, 351-378 (1955)		
		

\bibitem[McL85]{McL85}
L.~D.~McLerran, 
{\it ``Eleven lectures on the Physics of the Quark-Gluon Plasma,''}, (FNAL/C--84/101-T) (1985)

\bibitem[NS84]{NS84}
Niemi, A. J.; Semenoff, G. W.
{\it ``Finite-Temperature Quantum Field Theory in Minkowski Space,''}
Annals of Physics {\bf  152}, 105-129 (1984) 

\bibitem[OS75a]{OS75a}
K. Osterwalder and R. Schrader,
\textit{``Axioms For Euclidean Green’s Functions"},
Commun. Math. Phys. 31 (1973) 83.

\bibitem[OS75b]{OS75b}
	K. Osterwalder and R. Schrader,
	\textit{``Axioms for Euclidean Green’s Functions. 2."},
	Commun. Math. Phys. 42 (1975) 281.

\bibitem[Pa92]{Parwani}
R.~R.~Parwani
{\it ``Resummation in a hot scalar field theory,''} 
Phys.\ Rev.\ {\bf  D 45}, 4695 (1992) 


\bibitem[Ra96]{Radzikowski} 
  M.~J.~Radzikowski,
{\it ``Micro-local approach to the Hadamard condition in quantum field theory on curved space-time,''}
  Commun.\ Math.\ Phys.\  {\bf 179}, 529 (1996)
  
\bibitem[PSW04a]{ProkopecA}  
T.~Prokopec, M.~G.~Schmidt, and S.~Weinstock, 
{\it ``Transport equations for chiral fermions to order h bar and electroweak baryogenesis. Part 1,''} 
Annals Phys. {\bf 314}, 208-265 (2004) 
\bibitem[PSW04b]{ProkopecB}  
T.~Prokopec, M.~G.~Schmidt, and S.~Weinstock, 
{\it ``Transport equations for chiral fermions to order h-bar and electroweak baryogenesis. Part II,''}
Annals Phys. {\bf 314}, 267-320 (2004) 
  
  
\bibitem[St71]{Steinmann}
O.~Steinmann,
{\it ``Perturbation Expansions in Axiomatic Field Theory,''}
Lect. Notes in Phys. {\bf 11}. Berlin: Springer-Verlag, (1971)
 
\bibitem[Sch61]{Sch61}
J.~Schwinger,
{\it ``Brownian Motion of a Quantum Oscillator,''}
J.\ Math. Phys. {\bf 2}, 407 (1961)
 
 
\bibitem[vH86]{vH86}
van Hove, L. in: 
Editors. J.~de~Boer, E.~Dal and O.~Ulfbeck 
{\it ``The Lesson of Quantum Theory,''}, 
(North-Holland, Amsterdam, 1986)
  
\end{thebibliography}
\end{document}